\documentclass[a4paper]{article}

\usepackage[utf8]{inputenc}
\usepackage{a4wide}
\usepackage{amssymb,amsmath,amsthm}
\usepackage[colorlinks=true]{hyperref}
\usepackage[usenames,dvipsnames]{xcolor}
\usepackage{todonotes}
\usepackage{appendix}

\numberwithin{equation}{section}

\theoremstyle{plain}

\newtheorem{definition}{Definition}[section]

\newtheorem{Theorem}[definition]{Theorem}
\newtheorem{Proposition}[definition]{Proposition}
\newtheorem{Lemma}[definition]{Lemma}

\theoremstyle{remark}

\newcommand{\R}{\mathbb R}
\newcommand{\N}{\mathbb N}

\newcommand{\dd}{{\rm{d}}}

\newcommand{\im}{{\mathrm{i}}}
\newcommand{\Up}{{\mathcal U}_+}
\newcommand{\up}{u_+}
\newcommand{\lp}{\lambda_+}
\newcommand{\G}{{\mathcal{G}}}
\newcommand{\U}{{\mathcal{U}}}
\newcommand{\V}{{\mathcal{V}}}

\newcommand{\h}{{\mathcal{H}}}

\renewcommand{\d}{{\mathrm{d}}}

\newcommand{\GaG}{\Gamma_\G}
\newcommand{\ep}{\varepsilon}

\newcommand{\eps}{\varepsilon}

\usepackage[xcolor]{changebar}
\cbcolor{blue}

\newcommand{\X}{\mathfrak{X}}

\newcommand{\D}{\mathcal{D}}

\newcommand{\pt}{\partial}

\newcommand{\enumlabelformat}{\roman}

\newlength{\thelabelsep}
\newcounter{inlineenum}
\renewcommand{\theinlineenum}{\enumlabelformat{inlineenum}}

\let\epsilon\varepsilon
\let\phi\varphi

\newcommand{\Cinf}{\ensuremath{\mathcal{C}^\infty}}

\newcommand{\E}{\ensuremath{{\mathcal E}}}

\newcommand{\mb}[1]{\ensuremath{\mathbb{#1}}}

\newcommand{\Z}{\mb{Z}}

\newcommand{\EM}{\ensuremath{{\mathcal E}_{\mathrm{M}}}}

\newcommand{\NN}{\ensuremath{{\mathcal N}}}

\newcommand{\lara}[1]{\langle #1 \rangle}

\title{Cut-and-paste for impulsive gravitational waves with $\Lambda$:\\ The mathematical analysis} 
\author{
  Clemens S\"amann$^1$\thanks{{\tt clemens.saemann@maths.ox.ac.uk}},
  Benedict Schinnerl$^2$\thanks{{\tt b.schinnerl@gmail.com}},
  Roland Steinbauer$^2$\thanks{{\tt roland.steinbauer@univie.ac.at}}\ \ and Robert \v{S}varc$^3$\thanks{{\tt robert.svarc@mff.cuni.cz}} \\ \\
  $^1$ Mathematical Institute, University of Oxford,\\
  Andrew Wiles Building, Radcliffe Observatory Quarter, Woodstock Road, Oxford, UK.\\ \\
  $^2$ Faculty of Mathematics, University of Vienna, \\
  Oskar-Morgenstern-Platz 1, 1090 Vienna, Austria. \\ \\
  $^3$ Institute of Theoretical Physics,\\
  Charles University, Faculty of Mathematics and Physics, \\
  V Hole\v{s}ovi\v{c}k\'ach 2, 18000 Prague 8, Czech Republic.\\ \\ 
}
  
\begin{document}
\maketitle
\begin{abstract}
Impulsive gravitational waves are theoretical models of short but violent bursts of gravitational radiation. They are commonly described by two distinct spacetime metrics, one of local Lipschitz regularity, the other one even distributional. These two metrics are thought to be `physically equivalent' since they can be formally related by a `discontinuous coordinate transformation'. In this paper we provide a mathematical analysis of this issue for the entire class of nonexpanding impulsive gravitational waves propagating in a background spacetime of constant curvature. We devise a natural geometric regularisation procedure to show that the notorious change of variables arises as the distributional limit of a family of smooth coordinate transformations. In other words, we establish that both spacetimes arise as distributional limits of a smooth sandwich wave taken in different coordinate systems which are diffeomorphically related.
\bigskip

  \medskip
  
   
\end{abstract}

    


\newpage

\section{Introduction}

Impulsive gravitational waves are exact general relativistic spacetimes 
providing theoretic models of short but strong bursts of gravitational 
radiation. Originally introduced by R.\ Penrose (e.g.\ \cite{Pen:72}) they have 
since attracted the attention of researchers in exact spacetimes (who have 
widely generalised the original class of solutions, see e.g.\ \cite{P:02}), of
theoretical physicists (who have considered quantum scattering and the memory effect in these geometries as well as their astrophysical applications, see e.g.\ \cite{S:18,BH:03}), and of mathematicians (who have used them as relevant key-models in low regularity Lorentzian geometry). For a pedagogical introduction of the various constructions and models as well as their physical properties see \cite[Chapter 20]{GP:09}. 

Here we focus on a fundamental mathematical issue related to the fact that 
impulsive gravitational waves are metrics of low regularity. Indeed, they are 
commonly described by two `forms' of the metric, one (locally 
Lipschitz-)continuous, the other one distributional. The 
`physical equivalence' of these two descriptions has been established in several 
families of these models \emph{in a formal way}, leaving open some quite subtle 
problems in low regularity Lorentzian geometry. In fact, both `forms' of the 
metric are connected via a `discontinuous coordinate transformation' which 
reflects the Penrose junction conditions used to vividly construct these 
spacetimes in the first place. This `scissors and paste' approach \cite{Pen:72} 
was recently generalised to the $\Lambda\not=0$-case \cite{PSSS:19}.

In this work we completely solve these mathematical issues for the class of all 
\emph{nonexpanding} impulsive gravitational waves propagating on 
\emph{backgrounds of constant curvature}, i.e., Minkowski space and the (anti-)de Sitter 
universe. We do so by employing nonlinear distributional geometry, a method 
which is based on regularisation of the rough metrics and which also brings to 
light new insights into the geometry and the physics of these spacetimes.
We build upon the (nonlinear) distributional analysis of the geodesics in 
these spacetimes carried out in \cite{SSLP:16} and \cite{SS:17}. Thereby the current work concludes this long term research effort in a completely satisfactory way.

This article is organised in the following way. In the next section we recall the 
relevant aspects of the class of solutions we are working with. In particular, 
we discuss explicitly the `discontinuous coordinate transformation' thereby 
also fixing our notations and conventions. We also outline how our results 
and methods are related to previous works in the case $\Lambda=0$ \cite{KS:99}.
Then, in Section \ref{sec:ndaog}, we collect the necessary results from the nonlinear distributional analysis of the geodesics in the class of solutions at hand and fix our notation concerning nonlinear distributional geometry. Next, in Section \ref{sec:4} we study a special class of null geodesics, namely the null generators of (anti-)de Sitter space in a $5$-dimensional flat space with impulsive wave which provides us with a natural geometric regularisation of the notorious transformation. In Section \ref{sec:art} we analyse the regularised transformation and show that it is a generalised diffeomorphism in the sense of nonlinear distributional geometry. We close with a discussion of our results in Section \ref{sec:6}.

\section{Nonexpanding impulsive waves with $\Lambda$}\label{sec:2}
Here we describe the class of nonexpanding impulsive waves with an arbitrary value of $\Lambda$. These geometries have come into focus with the landmark work of Hotta and Tanaka \cite{HT:93}, where in analogy with the classical Aichelburg-Sexl approach \cite{AS:71}, the Schwarzschild--\emph{(anti-)de~Sitter} solution is boosted to ultrarelativistic speed to obtain a nonexpanding impulsive gravitational wave generated by a pair of null monopole particles. For an overview of the many more such solutions that have been found since, see e.g.\ \cite[Section 2]{PSSS:15}. 

\subsection{Metric representations and the `discontinuous transformation'}
In conformally flat coordinates these solutions for any value of $\Lambda$ 
take the \emph{distributional} form \cite{PG:99}
\begin{equation}\label{4D_imp}
  \dd s^2= \frac{2\,\dd\eta\,\dd\bar\eta-2\,\dd \U\,\dd 
    \V+2\h(\eta,\bar\eta)\,\delta({\U})\,\dd {\U}^2}
  {[\,1+\frac{1}{6}\Lambda(\eta\bar\eta-{\U}{\V})\,]^2} \,, 
\end{equation}
where $\h$ is a real-valued function and $\delta$ is the 
Dirac-distribution. 
Due to the occurrence of a distributional coefficient, \eqref{4D_imp} lies far beyond the Geroch-Traschen class of metrics \cite{GT:87}, which is defined by possessing Sobolev regularity $W^{1,2}_{\mbox{\small loc}}\cap L^\infty_{\mbox{\small loc}}$. It is known to be the largest class which allows one \emph{in general} to stably define the curvature in distributions (see also \cite{LFM:07,SV:09}). Nevertheless, due to its simple structure the curvature of \eqref{4D_imp} 
can be computed explicitly to give the impulsive Newman-Penrose components $\Psi_4=(1+\frac{1}{6}\Lambda\eta\bar\eta)^2\h_{,\eta\eta}\delta({\U})$, and $\Phi_{22}=(1+\frac{1}{6}\Lambda\eta\bar\eta)\big( (1+\frac{1}{6}\Lambda\eta\bar\eta) \h_{,\eta\bar\eta}+\frac{1}{6}\Lambda(\h-\eta \h_{,\eta}-\bar\eta \h_{,\bar\eta})\big)\delta({\U})$.
\medskip

The corresponding \emph{continuous} form of the metric is given by \cite{P:98,PG:99}
\begin{equation}\label{conti}
  \dd s^2= \frac{2\,|\dd Z+\up(h_{,Z\bar Z}\dd Z+{h}_{,\bar Z\bar Z}\dd\bar 
    Z)|^2-2\,\dd u\dd v}{[\,1+\frac{1}{6}\Lambda(Z\bar Z-uv+\up G)\,]^2}\,,
\end{equation}
where\footnote{This choice of sign of $G$ is in accordance with \cite{PO:01,SS:17,SSLP:16,PSSS:19}, 
  which are our main points of reference, but different e.g.\ from \cite{PSSS:15}.} ${G(Z,\bar Z)= Zh_{,Z}+\bar Zh_{,\bar Z}-h}$, and 
$h$ is a real valued function. Finally, $\up=\up(u)=0$ for $u\leq 0$ and 
$\up(u)=u$ for $u\geq0$ is the \emph{kink   function}. The metric \eqref{conti},
possessing a Lipschitz continuous coefficient, is of local Lipschitz 
regularity, which we denote by $C^{0,1}$. This is still beyond the reach of classical smooth Lorentzian geometry, which roughly reaches down to $C^{1,1}$ at least as far as convexity and causality is concerned \cite{Min:15,KSS:14,KSSV:14,GGKS:18}. However, the Lipschitz property is decisive since it prevents the most dramatic downfalls in causality theory which are known to occur for H\"older continuous metrics \cite{HW:51,CG:12,SS:18,GKSS:20,Min:19}. More specifically, in the context of the initial value problem for the geodesic equation, the Lipschitz property guarantees the existence of $C^{1,1}$-solutions \cite{Ste:14, LLS:21} which, due to the special geometry of the models at hand, are even (globally) unique \cite{PSSS:15}.
\medskip

A very useful way of thinking about the above metrics is the following. Starting with the conformally flat form of the constant curvature 
backgrounds, 
\begin{equation}\label{backgr}
  \d s_0^2= \frac{2\,\d\eta\,\d\bar\eta
    -2\,\d{\mathcal U}\,\d{\mathcal 
      V}}{[\,1+{\frac{1}{6}}\Lambda(\eta\bar\eta-{\mathcal U}{\mathcal
      V})\,]^2}
\end{equation}
we apply the transformation
\begin{equation}\label{ro:trsf}
  {\mathcal U}=u\,,\quad 
  {\mathcal V}=
  \begin{cases}
    v &\mbox{for ${\mathcal U}<0$}\\               
    v+h+uh_{,Z}h_{,\bar Z}  &\mbox{for ${\mathcal U}>0$}       
  \end{cases}
  \,,\quad
  \eta=
  \begin{cases}
    Z &\mbox{for ${\mathcal U}<0$}\\  
    Z+uh_{,\bar Z} &\mbox{for ${\mathcal U}>0$}  
  \end{cases}
\end{equation}
to \eqref{backgr} separately for negative and positive values of ${\mathcal 
  U}$ to formally obtain \eqref{conti}.  The corresponding 
distributional form \eqref{4D_imp} is \emph{formally} derived by writing 
\eqref{ro:trsf} in the form of a  `discontinuous coordinate transform' using the 
Heaviside function $\Theta$, i.e.,
\begin{equation}\label{trans}
   {\mathcal U}=u\,,\quad
   {\mathcal V}=v+\Theta\,h+\up\,h_{,Z}h_{,\bar Z}\,, \quad
  \eta=Z+\up\,h_{,\bar Z}\,.
\end{equation}
Then applying \eqref{trans} to \eqref{conti} and retaining all distributional 
terms one arrives at \eqref{4D_imp}. 

This transformation has first been given in \cite{Pen:72} for plane waves 
and in \cite{AB:97,PV:98} for the general \emph{pp}-wave case, i.e.,
nonexpanding impulsive waves propagating in a Minkowski background, hence 
$\Lambda=0$ in the above metrics \eqref{conti}, \eqref{4D_imp}.
Clearly, a mathematically sound treatment of the transformation \eqref{trans} 
is a delicate matter and it is the topic of this paper to completely clarify the situation.  


\subsection{Results in the pp-wave case}\label{sec:rpp}

A first rigorous result in this realm has been established in \cite{KS:99a} in the special case of impulsive \emph{pp}-waves. There, nonlinear distributional geometry \cite[Chapter 3]{GKOS:01} based on algebras of generalised functions \cite{Col:85} has been employed to show the following: The `discontinuous coordinate change' \eqref{trans} relating the distributional Brinkmann form of the metric, i.e., \eqref{4D_imp} with $\Lambda=0$ to the continuous Rosen form, i.e., \eqref{conti} with $\Lambda=0$ is the distributional limit of a \emph{generalised diffeomorphism}, a concept to be detailed below. Intuitively speaking this approach consists in viewing the impulsive wave as a limiting case of a sandwich wave with an arbitrarily regularised wave profile, where the two forms of the metric arise as its (distributional) limits taken in different coordinate systems. This result rests on two pillars:
\begin{enumerate}
 \item[(A)] The realisation that a special family of null geodesics in the 
distributional form of the metric precisely gives the coordinate lines of the 
coordinate system underlying the continuous form of the metric \cite{Ste:98}. 
In simpler words, the transformation \eqref{trans} is given by a special family of null geodesics.
\item[(B)] A fully nonlinear distributional analysis of the geodesics of the 
distributional metric. This is even a prerequisite to make (mathematical) sense of item (A): 
There is no valid solution concept for the geodesic equations of \eqref{4D_imp} 
with $\Lambda=0$ in classical distribution theory. Hence in \cite{Ste:98,KS:99} 
nonlinear distributional geometry has been employed to show existence, 
uniqueness and completeness\footnote{Observe that---although not stressed in the original works---especially the completeness result is remarkable, since it proves that the analyti\-cally `very singular' distributional spacetime is nonsingular in view of the standard definition \cite{HE:73}.} of geodesics.
\end{enumerate}

Here we set out to apply an analogous strategy to deal with the more 
involved $\Lambda\not=0$-case. In fact, building on earlier results we  provide the keystone of this approach: In \cite{SS:17} a nonlinear distributional analysis of the geodesic equation, see 
item (B) above, has been established. These results 
are in turn based on the formal analysis of the geodesics in \cite{PO:01} and the fixed point techniques put forward in \cite{SSLP:16}. We will collect the relevant statements in Section \ref{sec:nag}, below.

On the other hand, the geometric issue (A) has recently been resolved in \cite{PSSS:19} and we will review the results relevant for the present work in Secion \ref{sec:generators}, below. 
\medskip

In the nonlinear distributional analysis of the geodesics of \cite{SS:17,SSLP:16} it has, however, turned out that a five-dimensional approach is much better suited than a direct approach using the metric \eqref{4D_imp}. Indeed since \cite{PO:01} all works relevant for us have used this five-dimensional formalism and we close this section by briefly recalling it. The basic idea is to describe an impulsive wave in (anti-)de Sitter space as a hyperboloid in a five-dimensional flat space with impulsive wave \cite{PG:99a,PG:99}.


\subsection{The five-dimensional formalism}\label{sec:5D}



One starts out with the five-dimensional impulsive \textit{pp}-wave manifold
\begin{equation}\label{5D_imp}
  \dd s^{2}=-2\dd U \dd V+\dd Z_{2}^{2}+\dd Z_{3}^{2}+ \sigma\dd Z_{4}^{2} 
  +H(Z_{2},Z_{3},Z_{4})\delta(U)\dd U^{2} \,,
\end{equation}
with the constraint 
\begin{align}
  -2UV+Z_{2}^{2}+Z_{3}^{2}+ \sigma Z_{4}^{2} =\sigma a^{2} 
  \label{Constraint_Hyp} \,,
\end{align}
and parameters $\sigma=\pm 1 = \mathrm{sign}\, \Lambda$ and 
$a=\sqrt{\frac{3}{\sigma \Lambda}}$. The metric (\ref{5D_imp}) with 
(\ref{Constraint_Hyp}) thus represents an impulsive wave with the impulse located on the null hypersurface 
${U=0}$, 
\begin{equation}
  Z_{2}^2+Z_{3}^2+\sigma Z_{4}^2=\sigma a^2\,,\label{imp_surface}
\end{equation}
corresponding to a non-expanding 2-sphere for ${\Lambda>0}$ and a hyperboloidal 
2-surface for ${\Lambda<0}$. 
\medskip

Now, to relate (\ref{5D_imp}), (\ref{Constraint_Hyp}) 
to the four-dimensional distributional form (\ref{4D_imp}) we may use the  transformation
\begin{equation}
  {U} = \frac{\U}{\Omega}\,, \qquad   {V} = \frac{\V}{\Omega} \,, \qquad 
  Z_2+iZ_3= \frac{\sqrt2\,\eta}{\Omega}= \frac{x}{\Omega}+i \frac{y}{\Omega} 
\,, 
 \qquad Z_4 = a \left(\frac{2}{\Omega}-1\right), \label{CoordTrans_4D_to_5D}
\end{equation}
where we have used 
\begin{equation}
\Omega=1+{\textstyle\frac{1}{6}}\Lambda(\eta\bar\eta-\U\V)=1+{\textstyle\frac{1}{12}}\Lambda(x^2+y^2-2\U\V) \,,
\end{equation}
and the associated real coordinates $x,y$ with $\eta=1/\sqrt{2}(x+iy)$. Finally,  the profile functions are related by
\begin{equation}\label{HhRelation}
H =\frac{2\h}{1+\frac{1}{6}\Lambda\eta\bar\eta}
=\frac{2\h}{1+\frac{1}{12} \Lambda(x^2+y^2)}\,. 
\end{equation}

\section{Nonlinear distributional analysis of the geodesics}\label{sec:ndaog}

In this section we collect the results from the nonlinear distributional 
analysis of the geodesic equation in nonexpanding impulsive gravitational 
waves which we are going to use in the course of our work, cf.\ item (B) above. To keep this manuscript self contained, we start with a very terse review of the main elements of nonlinear distributional Lorentzian geometry.

\subsection{Nonlinear distributional geometry} \label{sec:ndg}

The theory we are going to summarize (for all details see 
\cite{KS:02a,KS:02b}, \cite[Section 3.2]{GKOS:01}) rests on J.F.\ Colombeau's 
construction of (so-called special) algebras of generalised functions 
\cite{Col:85}. These provide an extension of the linear theory of Schwartz 
distributions to the nonlinear realm retaining maximal consistency with 
classical analysis. The basic idea of the construction is regularisation of distributions via nets of smooth functions combined with asymptotic estimates in terms of a regularisation parameter. 

On a smooth (second countable and Hausdorff) manifold $M$ denote by 
$\E(M)$ the set of all nets of smooth functions $(u_\eps)_{\eps\in (0,1]=:I}$ 
which in addition depend smoothly\footnote{Smooth dependence on the 
parameter renders the theory technically more pleasant but was not assumed 
in earlier references, for details see \cite[Section 1]{BK:12}.} on $\eps$. The 
\emph{algebra of generalised functions on $M$} is defined 
as the quotient $\G(M) := 
\EM(M)/\NN(M)$ of \emph{moderate} modulo \emph{negligible} nets in $\E(M)$, 
which are defined via the following asymptotic estimates
\[
\EM(M) :=\{ (u_\eps)_\eps\in\E(M):\, \forall K\Subset M\
\forall P\in{\mathcal P}\ \exists N:\ \sup\limits_{p\in 
  K}|Pu_\eps(p)|=O(\eps^{-N}) \}\,,\\
\]
\[\NN(M)  :=\{ (u_\eps)_\eps\in\EM(M):\ \forall K\Subset M\
\forall m:\ \sup\limits_{p\in K}|u_\eps(p)|=O(\eps^{m}) \}\,.
\]
Here ${\mathcal P}$ denotes the space of all linear differential operators on $M$ and $K\Subset M$ means that $K$ is a compact subset of $M$. We write $u = [(u_\eps)_\eps]$ for the elements of $\G(M)$ and call $(u_\eps)_\eps$ a 
representative of the generalised function $u$. With sums, products, and the Lie 
derivative defined componentwise (i.e., for fixed $\eps$)  $\G(M)$ becomes a 
\emph{fine sheaf of differential algebras}. 

The space of distributions $\D'(M)$ can be embedded into 
$\G(M)$ via sheaf homomorphisms $\iota$ that preserve the product of 
$\Cinf(M)$-functions. A coarser way of relating generalised functions in 
$\G(M)$ to distributions is as follows: $u\in \G(M)$ is called 
\emph{associated} with $v\in \G(M)$, written $u\approx v$, if $u_\eps - v_\eps \to 0$
in $\D'(M)$. Moreover, $w\in \D'(M)$ is called associated with $u$ if 
$u\approx \iota(w)$.

\smallskip

More generally the space of \emph{generalised sections}  
of a vector bundle $E\to M$ is defined as
$
\label{tensorp} \GaG(M,E) = \G(M) \otimes_{\Cinf(M)} \Gamma(M,E)= 
L_{\Cinf(M)}(\Gamma(M,E^*),\G(M)).
$ 
It is a fine sheaf of finitely generated and projective $\G$-modules. 
For \emph{generalised tensor fields} of rank $r,s$ we use 
the notation 
\begin{equation}
  \G^r_s(M)\cong L_{\G(M)}(\G^0_1(M)^r,\G^1_0(M)^s;\G(M)). 
\end{equation}
Observe that it is possible to  insert generalised vector fields and one-forms 
into generalised tensors, which is not possible in the distributional setting, 
cf.\ \cite{Mar:67,deR:84}. This in turn allows one to work with generalised metrics 
much as in the smooth setting. Here a \emph{generalised pseudo-Riemannian 
metric} is a section $g\in\G^0_2(M)$ that is symmetric with determinant $\det g$ 
invertible in $\G$ (equivalently $|\det (g_\eps)_{ij}| \geq \eps^m$ for some $m$ on 
compact sets), and a well-defined index $\nu$ (the index of $g_\eps$ equals 
$\nu$ for $\eps$ small). By a `globalization Lemma' in \cite[Lemma 2.4, p.\ 
6]{KSSV:14} any generalised metric $g$ possesses a representative 
$(g_\eps)_\eps$ such that each $g_\eps$ is a smooth metric globally on $M$. We call a pair $(M,g)$ consisting of a smooth manifold and a generalised pseudo-Riemannian (Lorentzian) metric a \emph{generalised pseudo-Riemannian (Lorentzian) manifold}, and a \emph{generalised spacetime} if, in addition  to being Lorentzian, it can be time oriented by a smooth vector field. This setting consistently extends the `maximal distributional' one of Geroch and Traschen, 
see \cite{SV:09,S:08}. In particular, any generalised metric induces an isomorphism between generalised vector fields and one-forms, and there is a unique Levi-Civita connection $\nabla$ corresponding to $g$.
\medskip

Next, to speak of geodesics 
one uses the space of generalised curves $\G[J,M]$ \emph{taking values in 
$M$}, defined on an interval $J$. It is again a quotient of moderate modulo 
negligible nets $(\gamma_\ep)_\ep$ of smooth curves, where we call a net 
moderate (negligible) if $(\psi\circ \gamma_\eps)_\eps$ is moderate (negligible) for all smooth $\psi:M\to\R$. In addition, $(\gamma_\eps)_\eps$ is supposed to be \emph{c-bounded}, which means that $\gamma_\eps(K)$ is contained in a compact subset of $M$ for $\eps$ small and all compact sets $K\Subset J$. Observe that no distributional counterpart of such a space exists and it has long been realised that regularisation is a possible remedy, cf.\ \cite{Mar:67}.

The \emph{induced covariant derivative} of a generalised vector field 
$\xi=[(\xi_\eps)_\eps]\in\G^1_0(M)$ on a generalised curve $\gamma=[(\gamma_\eps)_\eps]\in\G[J,M]$ 
is defined componentwise (i.e., by the classical fomulae for fixed $\eps$) and 
gives again a generalised vector field $\xi'$ 
on $\gamma$. In particular, a \emph{geodesic} in a generalised 
spacetime is a curve $\gamma\in\G[J,M]$ satisfying 
$\gamma''=0$. Equivalently the usual local formula holds, i.e.,  
\begin{equation}\label{geo}
  \Big[\,\Big(\frac{d^2\gamma_\eps^k}{d\lambda^2}
  +\sum_{i,j}{\Gamma_\eps}^k_{ij}\frac{\gamma_\eps^i}{d\lambda}\frac{\gamma_\eps^j}{d\lambda}\Big)_\eps\,\Big]
  =0,
\end{equation}
where $\Gamma^k_{ij}=[({\Gamma_\eps}^k_{ij})_\eps]$ denotes the Christoffel symbols of 
the generalised metric $g=[(g_\eps)_\eps]$. 
Finally we say that a generalised spacetime  $(M,g)$ is \emph{geodesically 
  complete} if every geodesic $\gamma$ can be defined on $\R$ \cite[Definition 
2.1, p.\ 240]{SS:15}.

\subsection{Impulsive waves as generalised spacetimes \& the geodesic equation}
In this section we introduce the generalised metric form of nonexpanding impulsive waves for arbitrary values of $\Lambda$, using the five-dimensional formalism of Section \ref{sec:5D}. Indeed, starting with the metric \eqref{5D_imp} we replace the Dirac-delta with a generic regularisation: Choose any smooth function $\rho$ on $\R$ with unit integral and support in $[-1,1]$ and for $\eps\in(0,1]$ set $\delta_\eps(x):=(1/\eps)\,\rho(x/\eps)$. Such a net $(\delta_\eps)_\eps$ is called a \emph{model delta net} and we use it to define the regularised pseudo-Riemannian manifold $(\bar M=\R^5,\bar g_\eps)$ with line element    
\begin{equation}\label{5ipp}
\d\bar s_{\eps}^{2}=\d Z_{2}^{2}+\d Z_{3}^{2}+ \sigma \d Z_{4}^{2}-2\d U \d V
 +H(Z_{2},Z_{3},Z_{4})\delta_{\eps}(U) \d U^{2}\,.
\end{equation}  
Hence $(\bar M,\bar g_\eps)$ is a smooth \emph{sandwich wave} which is 
flat space outside the \emph{wave zone} given by $|U|\leq\eps$. The regularised 
impulsive wave spacetime of our interest $(M,g_\eps)$ is now given by the 
(anti-)de Sitter hyperboloid \eqref{Constraint_Hyp} embedded in $(\bar M,\bar g_\eps)$. 

To obtain an impulsive wave metric in $\G^0_2(\R^5)$ we use a \emph{model delta function}, that is an element $D\in\G(\R)$ that has a model delta net as a representative, $D=[(\delta_\eps)_\eps]$. Next we consider the $5$-dimensional generalised impulsive \emph{pp}-wave manifold 
$(\bar M=\R^5,\bar g)$ with
\begin{equation}\label{eq:Mbar}
  \d\bar s^2= dZ^2_2 + dZ^2_3 + \sigma dZ^2_4 - 2 dU dV + H(Z_2,Z_3,Z_4)D(U) 
dU^2\,.
\end{equation}
One easily checks that this defines a generalised metric with representative \eqref{5ipp}. At this point we specify the (A)dS hyperboloid $M$ in 
$(\bar M,\bar g)$ as usual, explcitly by 
\begin{align}\label{eq:M}\nonumber
  M := & \{(U,V,Z_2,Z_3,Z_4)\in\bar M : F(U,V,Z_2,Z_3,Z_4)=0\}\,,\quad 
\mbox{where}\\
&F(U,V,Z_2,Z_3,Z_4):=-2UV+Z_{2}^{2}+Z_{3}^{2}+ 
\sigma Z_{4}^{2}-\sigma a^{2}\,.
\end{align}
Note that $M$ is a (classical) smooth hypersurface. Finally, we restrict the metric $\bar g$ (again componentwise, that is for fixed $\eps$)  to $M$ to obtain the generalised spacetime $(M,g)$ which we take as our model of nonexpanding impulsive waves propagating in a(n 
anti-)de Sitter universe.
\medskip

To derive the geodesic equations in $(M,g$) we use the fact that in nonlinear generalised functions all classical formulae hold for fixed $\eps$. So we derive the $M$-geodesics from the condition that their $\bar M$-acceleration is normal to $M$, $\bar\nabla_T T=-\sigma g(T,\bar\nabla_T N)N/g(N,N)$. Here $\bar\nabla$ is the generalised Levi-Civita connection of $(\bar M,\bar g)$, and $T$ and $N$ denote the geodesic tangent and the (non-normalised) normal vector to $M$ defined via its representative $N^\alpha_\eps=g^{\alpha\beta}_\eps dF_\beta$, respectively. In this way we arrive at the geodesic equations 
for $\gamma=(U,V,Z_p)$:
\begin{align}\label{eq:geos:G}
  \ddot{U}
  &=-\Big( e + \frac{1}{2}\,\dot{U}^2\,\tilde{G}
  - \dot{U}\,\big(H\,D\,U\dot{\big)}\Big)\
  \frac{U}{\sigma a^2-U^2 H D}\,, \nonumber\\
  \ddot{V}-\frac{1}{2}\,H\,\dot{D}\,\dot{U}^2 - 
  \delta^{pq}H_{,p}\,\dot{Z}_q\,D\,\dot{U}
  &=-\Big(e + \frac{1}{2}\,\dot{U}^2\,\tilde{G}- \dot{U}\,\big( H\, D\, U 
  \dot{\big)}\Big)\
  \frac{V+H\,D\, U}{\sigma a^2-U^2 H D}\,,\nonumber\\
  \ddot{Z}_{i}-\frac{1}{2}H_{,i}\,D\, \dot{U}^2 &=-\Big(e + 
  \frac{1}{2}\,\dot{U}^2\,\tilde{G}
  - \dot{U}\,\big( H\, D\, U \dot{\big)}\Big)\ \frac{Z_{i}}{\sigma 
    a^2-U^2 H D}\,,\\
  \ddot{Z}_{4}-\frac{\sigma}{2}\,H_{,4}\,D\,\dot{U}^2 &=-\Big(e + 
\frac{1}{2}\,\dot{U}^2\,\tilde{G}
          - \dot{U}\,\big( H\, D\, U \dot{\big)}\Big)\
    \frac{Z_{4}}{\sigma a^2-U^2 H D}\,. \nonumber
\end{align}
Here $e=|\dot\gamma|=\pm 1,0$ for which it is natural to be fixed independently of $\eps$ and we have used the usual convention for spatial coordinates, i.e., $Z_p$ for $p=2,3,4$ and $Z_i$ for $i=2,3$. Moreover, we used the abbreviation $\tilde G = \delta^{p,q} Z_p H_{,q} - H$.
%

\subsection{Existence and uniqueness of geodesics}\label{sec:nag}

Next we briefly indicate how one proves unique solvability of the initial value problem for differential equations like \eqref{eq:geos:G} in generalised functions. This is basically done in three steps:
\begin{enumerate}
 \item[(1)] One proves existence of a so-called \emph{solution candidate}, in our case a net of smooth functions $\gamma_\eps=(U_\eps,V_\eps,Z_{p\eps}):J\to M$ depending smoothly on the parameter $\eps$ and solving the corresponding equation componentwise, i.e., for fixed (small) $\eps$. In our case this means $\gamma_\eps$ solves
\begin{align}\label{eq:geos}
\ddot{U}_\eps
&=-\Big( e + \frac{1}{2}\,\dot{U}_\eps^2\,\tilde{G_\eps}
           - \dot{U}_\eps\,\big(H
           \,\delta_{\eps}
           \,U_\eps\dot{\big)}\Big)\
           \frac{U_\eps}{\sigma a^2-U_\eps^2H\delta_\eps}\,, \nonumber\\
\ddot{V}_\eps-\frac{1}{2}\,H
           \,\delta^{'}_{\eps}
           \,\dot{U}_\eps^2 - \delta^{pq}H_{,p}
           \,\delta_{\eps}
           \,\dot{Z}_{q\eps}\,\dot{U}_\eps
&=-\Big(e + \frac{1}{2}\,\dot{U}_\eps^2\,\tilde{G}_\eps
          - \dot{U}_\eps\,
            \big( H\, \delta_\eps\, U_\eps \dot{\big)}\Big)\
   \frac{V_\eps+H\,\delta_{\eps}U_\eps}
        {\sigma a^2-U_\eps^2H\delta_\eps}\,,\nonumber\\
\ddot{Z}_{i\eps}-\frac{1}{2}H_{,i}\,\delta_{\eps}\dot{U}_\eps^2
&=-\Big(e + \frac{1}{2}\,\dot{U}_\eps^2\,\tilde{G}_\eps
          - \dot{U}_\eps\,
            \big( H\, \delta_\eps\, U_\eps \dot{\big)}\Big)\
   \frac{Z_{i\eps}}{\sigma a^2-U_\eps^2H\delta_\eps}\,,\\
\ddot{Z}_{4\eps}-\frac{\sigma}{2}\,H_{,4}\,\delta_{\eps}\dot{U}_\eps^2
&=-\Big(e + \frac{1}{2}\,\dot{U}_\eps^2\,\tilde{G}_\eps
          - \dot{U}_\eps\,
            \big( H\, \delta_\eps\, U_\eps \dot{\big)}\Big)\
    \frac{Z_{4\eps}}{\sigma a^2-U_\eps^2H\delta_\eps}\,, \nonumber
\end{align}
where we (again) have suppressed the parameter $\lambda$ as well as the dependencies on the
variables. However, note that always
\begin{align}\label{eq:ucomp}
 \delta_\eps&=\delta_\eps(U_\eps(\lambda))\,,\quad \delta'_\eps=\delta'_\eps(U_\eps(\lambda))\,,\nonumber \\
 \tilde G_\eps&=\tilde G_\eps\big(U_\eps(\lambda),Z_{p\eps}(\lambda)\big)\,,  \quad
  H=H\big(Z_{p\eps}(\lambda)\big)\,,\quad \text{and}\quad H_{,p}=H_{,p}\big(Z_{q\eps}(\lambda)\big)\,.
\end{align}
Observe that a solution candidate $(\gamma_\eps)_\eps$ actually is comprised of geodesics $\gamma_\eps$ of the regularised spacetime $(M,g_\eps)$, cf.\ \eqref{5ipp}.


\item[(2)] One shows \emph{existence} of a generalised solution by establishing c-boundedness and moderateness of the 
solution candidate, i.e., $\gamma:=[(\gamma_\eps)_\eps]\in\G[J,M]$.
\item[(3)] To show \emph{uniqueness} in $\G$ one solves a negligibly perturbed version of the equations---in our case \eqref{eq:geos}, with negligible nets added at the right hand side of every equation---and shows that the corresponding net of 
solution $(\tilde{\gamma}_\eps)_\eps$ only differs negligibly from $(\gamma_\eps)_\eps$, i.e., $[(\tilde{\gamma}_\eps)_\eps] =[(\gamma_\eps)_\eps]$. Observe that this amounts to an additional stability statement for the solutions of the regularised equation.
\end{enumerate}

With this let us turn to initial conditions for solutions of the system \eqref{eq:geos:G} appropriate for our purpose, see also Figure \ref{fig:seed1}. Consider a geodesic $\gamma^-=(U^-,V^-,Z^-_{p})$ of the background (an\-ti-)\-de Sitter universe \emph{without} impulsive wave but reaching $U=0$ and assume that we have chosen an affine parameter such that $U^-(0)=0$ and $\dot U^-(0)=1$. Further, since we will only be interested in null geodesics, {we have} that $\dot \gamma^-$ is {null, i.e.,} $e=0$. 
Now we conveniently prescribe initial data at the affine parameter value $\lambda=0$, 
\begin{equation}\label{eq:data0}
 \gamma^-(0)=(0,V^0,Z^0_{p})\,,\qquad \dot \gamma^-(0)=(1,\dot V^0,\dot Z^0_{p})\,,
\end{equation}
where the constants satisfy the constraints
\begin{equation}\label{const1}
(Z^0_{2 })^{2}+(Z^0_{3 })^{2}+ \sigma (Z^0_{4 })^{2}=\sigma a^{2},\quad
Z^0_{2 } \dot Z^0_{2 }+Z^0_{3 } \dot Z^0_{3 }+ \sigma Z_{4 }^0\dot  Z_{4 }^0-V^0=0\,,
\end{equation}
and the normalization
\begin{equation}\label{eq:norm}
 -2 \dot V^0 +(\dot Z^0_{2 })^2+(\dot Z^0_{3 })^2+\sigma(\dot Z^0_{4 })^2=e=0\,.
\end{equation}
We will refer to $\gamma^-$ as \emph{seed geodesics} and start to think of it as geodesics in the impulsive wave spacetime (\ref{5D_imp}), (\ref{Constraint_Hyp}) \emph{`in front' of the impulse}, that is for $U^-<0$. Also, $\gamma^-$ is a geodesic in the regularised spacetime \eqref{5ipp}, \eqref{Constraint_Hyp} \emph{`in front' of the sandwich wave}, that is for $U^-\leq -\eps$. We will denote the affine parameter time when $\gamma^-$ enters this regularisation wave region by $\alpha_\eps$, i.e.,
\begin{equation}
 U^-(\alpha_\eps)=-\eps\,.
\end{equation}
Observe that $\alpha_\eps\to 0$ from below as $\eps\to 0$. 
Finally, we come to setting up the data for the solution candidate $\gamma_\eps$ of 
the system \eqref{eq:geos:G} by
\begin{equation}\label{eq:real-data}
 \gamma_\eps(\alpha_\eps)=\gamma^-(\alpha_\eps),\quad
 \dot \gamma_\eps(\alpha_\eps)=\dot\gamma^-(\alpha_\eps),
\end{equation}
i.e., as the data the seed geodesic assumes at $\alpha_\eps$. We 
will frequently refer to these data \eqref{eq:real-data} as initial data 
constructed from the seed geodesic $\gamma^-$ with data \eqref{eq:data0}.
\begin{figure}[h]
\begin{center}
\def\svgwidth{0.65\textwidth}
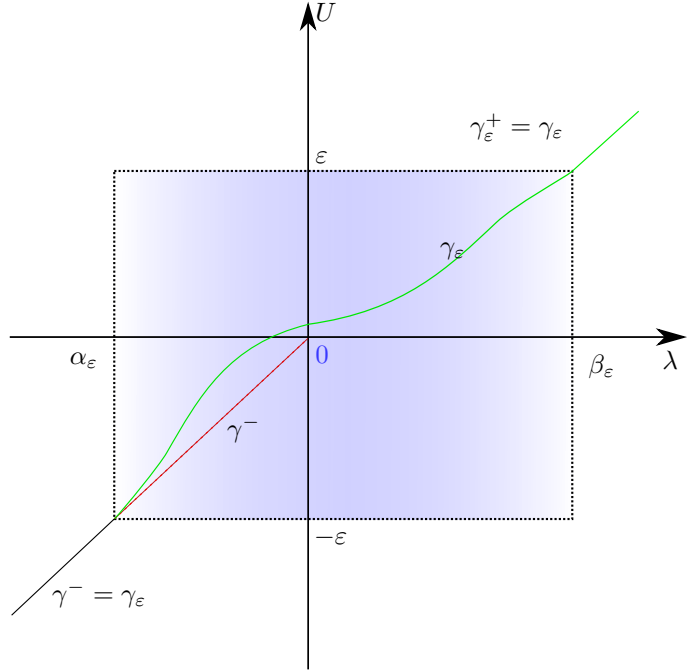
\caption{The $U$-component of the seed geodesic $\gamma^-$
is depicted in black until it reaches the regularisation sandwich at parameter
time $\lambda=\alpha_\eps$, i.e., $U^-(\alpha_\eps)=-\eps$. While in the background spacetime it would continue as the dotted red line to $U=0$ at $\lambda=0$, in the regularised spacetime it continues as a solution $\gamma_\eps$ of \eqref{eq:geos} with data \eqref{eq:real-data} (depicted in green). 
Theorem \ref{thm:ex+un} guarantees that $\gamma_\eps$ (for $\eps$ small) leaves the regularisation sandwich at $\lambda=\beta_\eps$ and continues as a background geodesic.}
\label{fig:seed1}
\end{center}
\end{figure}

The central result on the solvability of the geodesic equations for the generalised spacetime \eqref{eq:Mbar}, \eqref{eq:M} now is the following, cf.\ \cite[Thms.\ 3.6, 3.7]{SS:17}. 
\begin{Theorem}[Global existence and uniqueness]\label{thm:ex+un}
The geodesic equation \eqref{eq:geos:G} with initial data \eqref{eq:real-data} (constructed from the seed geodesic $\gamma^-$ with data \eqref{eq:data0}) possesses a global unique solution $\gamma=(U,V,Z_p)=[(\gamma_\eps)_\eps]\in \G[\R,M]$.    
\end{Theorem}

\subsection{Associated geodesics}\label{sec:ass}


Here we recall the associated geodesics of the solutions $[(\gamma_\eps)_\eps]$ of Theorem \ref{thm:ex+un} which were given in \cite[Sec.\ V]{SS:17} based on the explicit calculations in \cite[Sec.\ 5, Appendix B]{SSLP:16}, see also \cite[Eqs.\ (38), (39)]{PO:01} for a formal approach. These are calculated as the (distributional) limits of the representatives $\gamma_\eps=(U_\eps,V_\eps,Z_{p\eps})$ of the solutions of Theorem \ref{thm:ex+un}.


Now to formulate a precise result we first establish a notation for the limiting geodesics. Clearly in front of the impulse, that is for $\lambda<0$ corresponding to $U_\eps<0$, $\gamma_\eps$ will converge to the seed geodesic $\gamma^-$. Similarly, behind the impulse, that is for $\lambda>0$ corresponding to $U_\eps>0$, $\gamma_\eps$ will also converge to a geodesic $\gamma^+$ of the background (A)dS space. Here $\gamma^+=(U^+,V^+,Z_p^+)$ is specified by the values of $\gamma_\eps$ and $\dot\gamma_\eps$ upon leaving the regularisation zone: Indeed it is shown in the course of the proof of Theorem \ref{thm:ex+un} (cf.\ \cite[eq.\ (34) and below]{SS:17}) that there is a parameter value $\beta_\eps>0$ such that $U_\eps(\beta_\eps)=\eps$ and that $\beta_\eps\searrow 0$ for $\eps\to 0$, c.f.\ \cite[Lem.\ A2]{SSLP:16}. Moreover, the corresponding values of $\gamma_\eps(\beta_\eps)$ and $\dot\gamma_\eps(\beta_\eps)$ converge (cf.\ \cite[Prop.\ 5.3]{SSLP:16}). More precisley, we have
\begin{align}\nonumber
    \gamma_\eps(\beta_\eps)
    =\Big(U_\eps(\beta_\eps),V_\eps(\beta_\eps),Z_{p\eps}(\beta_\eps)\Big)&\to\ (0,B+V^0,Z^0_p), \\
    \dot\gamma_\eps(\beta_\eps)
    =\Big(\dot U_\eps(\beta_\eps),\dot V_\eps(\beta_\eps),\dot Z_{p\eps}(\beta_\eps)\Big)
    &\to\ (1,C+\dot V^0,A_p+\dot Z^0_p),
\end{align}
where $V^0$, $Z^0_{p}$, $\dot V^0$, $\dot Z^0_{p}$ are the corresponding seed data \eqref{eq:data0},  and
\begin{align}\label{eq:ro:consistency}\nonumber
 {A_i} 
 &=\frac{1}{2}\, 
         \Bigl(\ H_{,i}(Z_r^0) + \frac{Z_i^0}{\sigma a^2}\big(H(Z_r^0)
          - \delta^{pq}Z_p^0H_{,q}(Z_r^0) \big)\Bigr)\,,\\
 {A_4} 
 &=\ \frac{1}{2}\, 
         \Bigl(\sigma H_{,4}(Z_r^0) + \frac{Z_4^0}{\sigma a^2}\big(H(Z_r^0) -
 \delta^{pq}Z_p^0H_{,q}(Z_r^0) \big)\Bigr)\,, \nonumber\\
 {B}
 &=\  \frac{1}{2}H(Z^0_p) \,,\\
\nonumber
 {C}
&= \frac{1}{8} \bigg(H_{,2}(Z_r^0)^2 + H_{,3}(Z_r^0)^2 +\sigma H_{,4}(Z_r^0)^2 
+\frac{1}{\sigma a^2} H(Z_r^0)^2 - \frac{1}{\sigma a^2}\left(\delta^{pq}Z_p^0
H_{,q}(Z_r^0)\right)^2\bigg)\\\nonumber
&\qquad\quad  -\frac{1}{2 \sigma
a^2}\left(\delta^{pq}Z_{p}^0H_{,q}(Z_r^0) -
H(Z_r^0)\right)V^0 + \frac{1}{2}\delta^{pq}H_{,p}(Z_r^0)\dot{Z}_q^0\,. 
\end{align}
%
Now we explicitly fix the limiting geodesic behind the wave $\gamma^+$ by prescribing the data
\begin{equation}\label{eq:Xdata_ro}
\gamma^+(0)=(0,{B}+V^0,Z_p^0)\,,\quad \text{and}\quad
\dot \gamma^+(0)=(1,{C}+\dot V^0,{A_p}+\dot Z^0_p)\,,
\end{equation}
and denote the corresponding global limiting geodesic by
\begin{equation}\label{eq:limgeos}
 \tilde\gamma(\lambda)=(\tilde U,\tilde V,\tilde Z_p)(\lambda):=\begin{cases}
  \gamma^-(\lambda)\,,\qquad&\lambda \leq 0\\
  \gamma^+(\lambda)\,,\qquad&\lambda>0\,.\end{cases}
\end{equation}
Then we have by \cite[Thm.\ 5.2]{SS:17} and \cite[Thm.\ 5.1]{SSLP:16}, respectively, the following convergence result.
\begin{Theorem}[Associated geodesics]\label{thm:ag}
The solution $\gamma=(U,V,Z_p)=[(\gamma_\eps)_\eps]$ of Theorem \ref{thm:ex+un} is associated with the limiting geodesic $\tilde \gamma$ of \eqref{eq:limgeos}. Moreover we have $U\approx_1\tilde U$ and $Z_p\approx_0\tilde Z_p$. 
\end{Theorem}

Recall that e.g.\ $V=[(V_\eps)_\eps]\approx \tilde V$ means that 
$\lim_{\eps\to0} \int_\R V_\eps(x)\phi(x)\,\dd x = \lara{\tilde V,\phi}$ for all test functions $\phi\in\D(\R)$ (and $\langle \,.\,,\,.\,\rangle$ denotes the distributional action). Similarly, $\tilde U=[(U_\eps)_\eps]\approx_k \tilde U$ $(k\in\N)$ means that
$U_\eps\to \tilde U$ in $C^k(\R)$, i.e., uniformly on all compact subsets of $\R$ up to derivatives of order $k$. Note that the convergences given by Theorem \ref{thm:ag} are optimal in light of $\tilde V$ and $\dot{\tilde{Z}}_p$ being discontinuous across $\lambda=0$, i.e., the limiting geodesics being refracted geodesics of the background suffering a jump in the $V$-position and $V$-velocity as well as in the $Z_p$-velocity, cf.\ \cite[Section 5]{SSLP:16}.  

Note that the limiting geodesics $\tilde\gamma$ of \eqref{eq:limgeos} can be interpreted as the geodesics of the distributional spacetime \eqref{5D_imp}, \eqref{Constraint_Hyp}. Keep in mind, however, that \eqref{eq:limgeos} \emph{does not} solve the (formal) geodesic equations of the distributional spacetime (see \cite[Eq.\ (2.6)]{SSLP:16}, \cite[Eq.\ (28)]{PO:01}) by the lack of a consistent solution concept. Indeed the low regularity of $\tilde\gamma$ does not allow one to insert it into these equations. 

Finally, we will mainly be interested in the null case where we can give a more transparent form of $\tilde \gamma$, cf.\ \cite[Eq.\ (29)]{PSSS:19}
due to the fact that the null geodesics of the (A)dS background are just straight lines, cf.\ \cite[Sec.\ 4]{PO:01}
\begin{equation}\label{5DNullGeod}
\tilde \gamma(\lambda) = \left(
\begin{array}{c}
\lambda\\
V^0+\dot{V}^0\lambda+\Theta(\lambda)B+C\lp\\
Z_p^0+\dot{Z}^0_p\lambda+A_p\lp
\end{array}\right)\,.
\end{equation}

\section{The null geodesic generators and the transformation}\label{sec:4}

In this section we turn to issue (A) of Section \ref{sec:rpp} which has been resolved in the $\Lambda\not=0$-case in \cite{PSSS:19}. Briefly, the main result is that the null geodesic generators of the (A)dS hyperboliod give rise to the notorious transformation \eqref{trans}. We will combine this insight with the results of Section \ref{sec:ndaog} to derive a geometric regularisation of the transformation.

\subsection{The null geodesic generators \& the `discontinuous transformation'}\label{sec:generators}

To begin with we relate the limiting null geodesics of \eqref{5DNullGeod} to the null geodesic generators of the (A)dS hyperboloid. The latter are most conveniently found using the conformally flat coordinates of the (A)dS background \eqref{backgr} to be (cf.\ \cite[Eq.\ (18), (19)]{PSSS:19}\footnote{Where we have already set $c=1$, see item (3) in \cite[Subsec.\ III.B]{PSSS:19}.})
\begin{equation}
\gamma^g_{4D}({\lambda})\ =\ \left(
\begin{array}{c}
     \U^g(\lambda)  \\ \V^g(\lambda) \\ x^g(\lambda) \\ y^g(\lambda) 
\end{array}
\right)\ =\ 
\left(
\begin{array}{c}
\alpha^2\lambda/(1-\beta\lambda)\\
\V_0\\
x_0\\
y_0\\
\end{array}\right)\,, \label{4DGeodAfBetaNeq0}
\end{equation}
where 
\begin{equation}
\alpha= 1+\frac{\Lambda}{12}\left(x_0^2+y_0^2\right) \,, \qquad
\beta= -\frac{\Lambda}{6}\,\V_0 \,. \label{AlphaBetaDef}
\end{equation}
This family of null geodesics is parameterised by three real constants  fixing the positions at the parameter value $\U=\lambda=0$, i.e.,
$\gamma^g_{4D}(0)=(0,\V_0,x_0,y_0$).

Next we write the null generators \eqref{4DGeodAfBetaNeq0} in the five-dimensional representation of Section \ref{sec:5D} but still parameterised by the 4D-data $(\V_0, x_0,y_0)$ (cf.\ \cite[Eq.\ (26)]{PSSS:19})
\begin{equation}\label{5DGeodFin}
\gamma^g_{5D}(\lambda)\ =\ 
\left(\begin{array}{c}
    U^g(\lambda)\\ V^g(\lambda) \\ Z^g_2(\lambda)\\ Z^g_3(\lambda)\\ Z^g_4(\lambda)
\end{array}
\right)\ =\ 
\frac{1-\beta\lambda}{\alpha}\ \left(
\begin{array}{c}
  \alpha\lambda/(1-\beta\lambda)\\
\V_0\\
x_0\\
y_0\\
a\Big(2-\alpha/(1-\beta\lambda)\Big)
\end{array}\right)\,.
\end{equation}
Observe from the first line that $U^g(\lambda)=\lambda$. Now we use the geodesics
\eqref{5DGeodFin} for $\lambda\leq 0$ as seed for the global limiting null geodesics \eqref{5DNullGeod}, that is, according to \eqref{CoordTrans_4D_to_5D}, we set the eight constants $V^0$, $Z^0_p$, $\dot V^0$, and $\dot Z^0_p$ of \eqref{eq:data0} to 
\begin{align}
&V^0=\frac{\V_0}{\alpha} \,,& &Z_2^0=\frac{x_0}{\alpha} \,,&
&Z_3^0=\frac{y_0}{\alpha} \,,& &Z_4^0=a\left(\frac{2}{\alpha}-1\right) \,,
\nonumber \\
&\dot{V}^0=-\frac{\beta}{\alpha}\V_0 \,,&
&\dot{Z}^0_2=-\frac{\beta}{\alpha}x_0 \,,& &\dot{Z}^0_3=-\frac{\beta}{\alpha}y_0
\,,& &\dot{Z}^0_4=-2a\frac{\beta}{\alpha} \,, \label{GenVel}
\end{align}
which relates them to the three parameters $\V_0,x_0,y_0$.
Now we obtain the global limiting geodesic \eqref{5DNullGeod} with seed given by the null geodesic generator of the (A)dS hyperboloid  with data $(\V_0,x_0,y_0)$ as
\begin{equation}
    \gamma_{5D}[\V_0,x_0,y_0](\lambda)=\left(
\begin{array}{c}
\lambda\\
V^0+\dot{V}^0\lambda+\Theta(\lambda)B+C\lp\\
Z_p^0+\dot{Z}^0_p\lambda+A_p\lp
\end{array}\right)\,,
\end{equation}
where we have to substitute \eqref{GenVel} into \eqref{eq:ro:consistency}.
Finally, we express these geodesics in the 4D coordinates $(\U,\V,x,y)$ of \eqref{4D_imp} (cf.\ \cite[Eq.\ (40)]{PSSS:19}) as
\begin{align}
\!\!\!\!\gamma_{4D}[\V_0,x_0,y_0](\U)=\left(
\begin{array}{c}
\U \vspace{2.0mm} \\
\V_0+\Theta(\U)\,\h^\im+\Up
\,\frac{1}{2}\big((\h_{,x}^\im)^2+(\h_{,y}^\im)^2\big) \vspace{2.0mm}  \\
x_0^j+\Up\h_{,j}^\im\\
\end{array}\right)\!, \label{Global4DNullGeod_uParam}
\end{align}
where the profile function and its derivatives are explicitly related by, see (\ref{HhRelation})\footnote{Here we use the relations ${\Lambda=3\sigma/a^2}$ and ${(x_0)^2+(y_0)^2=4\sigma a^2(\alpha-1)}$.}
\begin{align}
H^\im_{,j}=2\h_{,j}^\im-\h^\im\,\frac{x_0^{j}}{\sigma\alpha a^2} \,, \qquad
H^\im_{,4}=-\frac{1}{a}\left(x_0\h^\im_{,x}+y_0\h^\im_{,y}\right)+2\h^\im\,\frac
{\alpha-1}{\alpha a}\,. \label{DerivativesHh}
\end{align}
Here $H^\im$ and ${\h}^\im$ as well as the corresponding derivatives denote the respective values at the instant of interaction of the geodesics with the impulse, i.e., at the parameter value $\U=\lambda=0$. So we e.g.\ have  $H_{,j}^\im=H_{,j}(Z^0_p)$. 
Also the constants $A_p$, $B$, and $C$ can explicitly be written in terms of $\h$, cf.\ \cite[Eqs.\ (33)--(35)]{PSSS:19}
\begin{align}
& A_j=\h^\im_{,j}+\frac{x_0^j}{2\sigma\alpha a^2}\, \G, \quad 
A_4=\frac{1}{\sigma\alpha a} \G, \quad 
B = \frac{1}{\alpha}\h^\im \,,\nonumber \\&
C
=\frac{1}{2}\big((\h^\im_{,x})^2+(\h^\im_{,y})^2\big)+\frac{1}{
2\sigma\alpha a^2}\big((\h^\im+\G)\V_0+\h^\im\G\big) \,, \label{5DCoeffSpec}
\end{align}
where $\G$ and the conformal factor 
take the form
\begin{equation}
\G \equiv \h^\im-x_0\h^\im_{,x}-y_0\h^\im_{,y} \quad\text{and}\quad 
\Omega=
\frac{\alpha}{1-\beta\lambda+\frac{\Lambda}{6}\G\,\lp}\,.
\end{equation}


The key observation at this point is that the limiting geodesics \eqref{Global4DNullGeod_uParam} exactly match the transformation \eqref{trans}. More precisely (cf.\ \cite[Sec.\ IV]{PSSS:19}), we may employ \eqref{Global4DNullGeod_uParam} to transform the coordinates $(u,v,Z)\equiv(u,v,X,Y)$ in which the metric is \emph{continuous} (cf.\ \eqref{conti}) to the coordinates $(\U,\V,\eta)\equiv(\U,\V,x,y)$ in which the metric is \emph{distributional} (cf.\ \eqref{4D_imp}) via
\begin{equation}\label{eq:final}
  \left(\begin{array}{c}
    u\\v\\X\\Y
   \end{array}\right)\,
   \mapsto\, \gamma_{4D}[v,X,Y](u)=
   \left(\begin{array}{c}
     u\\v+\Theta(u)\,\h^\im+\up
     \,\frac{1}{2}\big((\h_{,X}^\im)^2+(\h_{,Y}^\im)^2\big) \\
     X+\up\h_{,X}^\im\\ Y+\up\h_{,Y}^\im
   \end{array}\right)
   =
   \left(\begin{array}{c}
     \U\\\V\\x\\y
   \end{array}\right)\,.
\end{equation}
We have hence formally recovered the `discontinuous transformation' from a special family of global limiting null geodesics, which can be interpreted as the geodesics of the distributional spacetime  \eqref{5D_imp}, \eqref{Constraint_Hyp}, cf. the penultimate paragraph of Section \ref{sec:ass}.

Moreover, the behaviour of these geodesics can be vividly depicted, see Figure \ref{fig:nullgenerators} (cf.\ \cite[Fig.\ 2]{PSSS:19}, and \cite[Fig.\ 6]{PS:22}) in a way that directly generalizes Penrose's original illustration for the $\Lambda=0$-case in \cite[Fig.\ 2]{Pen:72}:  The null geodesic generators of (A)dS starting in the `lower half' $(A)d\mathcal{S}^-$ (i.e., for $U=\lambda<0$) due to their interaction with the wave impulse do not continue as unbroken null generators into $(A)d\mathcal{S}^+$ (indicated by the dashed line in the upper left parts). Rather they jump at $\{U=\lambda=0\}$, cf.\ the $\Theta$-term in \eqref{Global4DNullGeod_uParam} (hence in \eqref{eq:final}) but also get refracted, cf.\ the $\lambda_+$-terms, to become the appropriate null generators of $(A)d\mathcal{S}^+$.

\begin{figure}[h]
\centering
\begin{minipage}{.45\textwidth}
  \centering
  \includegraphics[width=\linewidth]{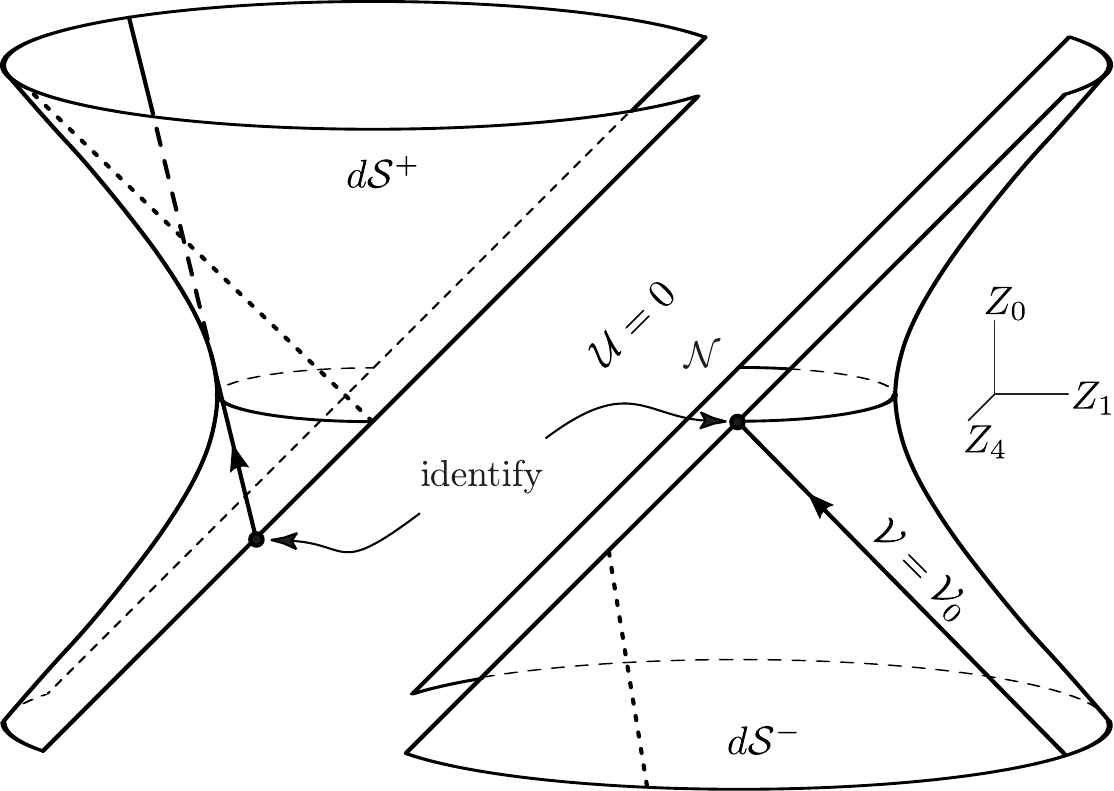}
\end{minipage}%
\hfill
\begin{minipage}{.45\textwidth}
  \centering
  \includegraphics[width=.7\linewidth]{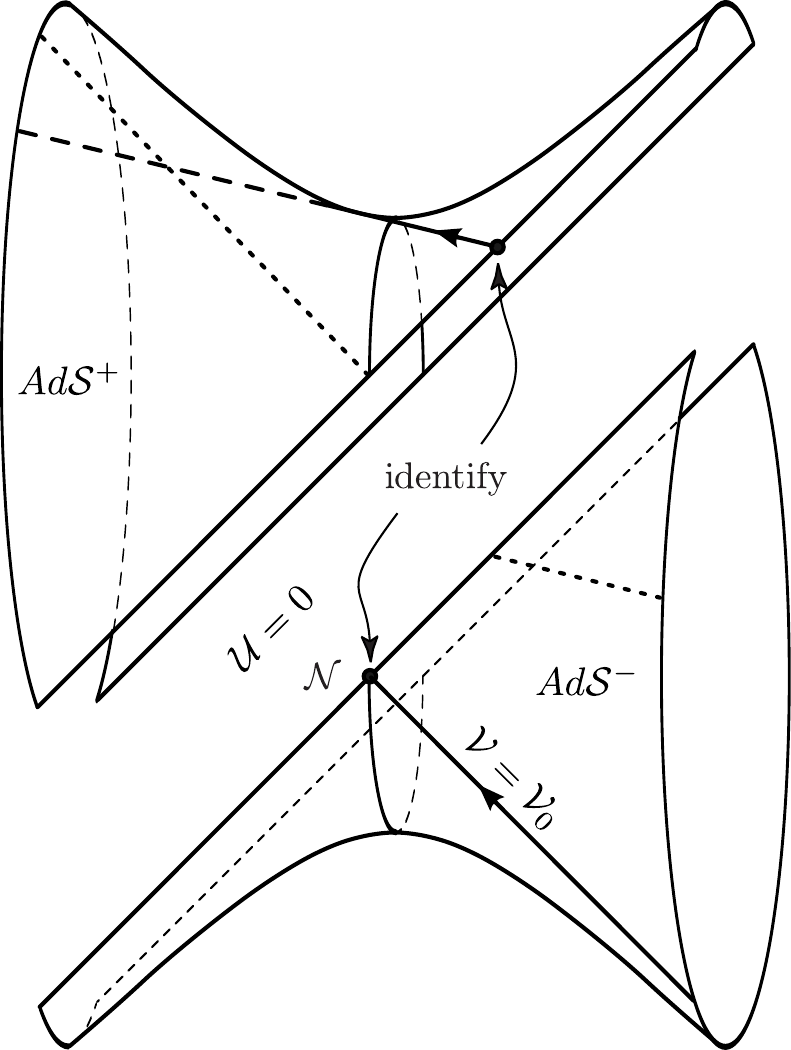}
\end{minipage}
\caption{The null geodesic generators of (A)dS and their interaction with the impulsive wave.}
\label{fig:nullgenerators}
\end{figure}

The upshot is that these `broken geodesic generators' basically are the coordinate lines of the \emph{coordinate system in which the metric becomes continuous}, i.e., \eqref{conti}. But we do not only have these limiting geodesics at hand but also the generalised geodesics of Theorem \ref{thm:ex+un}. This will allow us to geometrically regularise the transformation, which we will explicitly do next.

\subsection{The geometrically regularised transformation} 
Using the ideas laid out above we now give the explicit form of the transformation in nonlinear generalised functions. To begin with, note that we only have carried out the nonlinear distributional analysis of the geodesics in the $5D$-form\footnote{As explained in  \cite[p.\ 3]{PSSS:16}, the geodesic equations in the $4D$-form are even wilder and a direct approach seems to be out of reach.}. Therefore we split up the transformation in the following way: starting from the `continuous' $4D$-coordinates $(u,v,Z)\equiv(u,v,X,Y)$ of \eqref{conti}  we first use the transformation \eqref{CoordTrans_4D_to_5D}\footnote{Observe that \eqref{CoordTrans_4D_to_5D} literally transforms the conformally flat coordinates $({\mathcal{U},\mathcal{V},\eta})$ to the $5D$-coordinates. But here we transform data `in front' of the wave, where $({\mathcal{U},\mathcal{V},\eta})$ are trivially related to $(u,v,Z)$, cf.\ \eqref{ro:trsf} and so in \eqref{CoordTrans_4D_to_5D} we have to replace $({\mathcal{U},\mathcal{V},\eta})$ with  $(u,v,Z)$.} to go to the $5D$-coordinates $(U,V,Z_p)$. Then we use the regularised geodesics to transform 
\begin{equation}\label{eq:regtrsf}
    (U,V,Z_p)\mapsto T_\eps(U,V,Z_p):=\gamma^\eps_{5D}[V,Z_p](U)=:(\bar U_\eps,\bar V_\eps,\bar Z_{p\eps})\,,
\end{equation} 
where\footnote{We here write the regularisation parameter as a superscript rather than as a subscript to make the notation more appealing.} $[(\gamma^\eps_{5D}[V,Z_p])_\eps]$ are the generalised solutions of the geodesic equations provided by Theorem \ref{thm:ex+un} with data
constructed from the seed geodesic $\gamma^g_{5D}$ of \eqref{5DGeodFin}, i.e., the null generator with data $(V,Z_p,\dot V,\dot Z_p)$ as in \eqref{GenVel},  but now derived from $(v,X,Y)$ instead of $(\V_0,x_0,y_0)$. We will specify this data explicitly below but first we turn to the final final part of the transformation. 
For this we use the inverse of \eqref{CoordTrans_4D_to_5D}, i.e.,
\begin{equation}
\U=\Omega \bar U \,, \qquad   \V=\Omega \bar V \,, \qquad x=\Omega \bar Z_2 \,, \qquad y=\Omega
\bar Z_3 \,, \qquad \hbox{with} \qquad \Omega=\frac{2a}{\bar Z_4+a} \,,
\label{CoordTrans_5D_to_4D}
\end{equation}
componentwise (that is for fixed $\eps$), to go from the $5D$-coordinates $(\bar U_\eps,\bar V_\eps,\bar Z_{p\eps})$ to  $4D$ coordinates $(\U_\eps,\V_\eps,x_\eps,y_\eps)$, which provide a regularisation of the $4D$ `distributional' system. That is, overall the transformation we are going to employ takes the form
\begin{align}\label{eq:trsf-overall}
    \left(\begin{array}{c}u\\v\\X\\Y\end{array}\right)\quad 
    \overset{\eqref{CoordTrans_4D_to_5D}}{\longmapsto}\quad
    \left(\begin{array}{c}U\\V\\Z_p\end{array}\right)\quad 
    \overset{T_\eps}{\longmapsto}\quad
    \gamma^\eps_{5D}[V,Z_p](U)\ = \ 
    \left(\begin{array}{c}\bar U_\eps\\\bar V_\eps\\\bar Z_{p_\eps}\end{array}\right)\quad 
    \overset{\eqref{CoordTrans_5D_to_4D}}{\longmapsto}\quad
    \left(\begin{array}{c}\U_\eps\\ \V_\eps\\ x_\eps\\ y_\eps\end{array}\right)\,.
\end{align}
This is the sensible geometric regularisation of the `discontinuous transformation' \eqref{eq:final}, which we have been aiming for.

Since the first and the third map in the overall transformation \eqref{eq:trsf-overall}  are (classical smooth) diffeomorphisms it is sufficient to restrict our nonlinear distributional analysis of the transformation to $T_\eps$. Therefore we do not need to take into account that the data $(V,Z_p,\dot V,\dot Z_p)$ of \eqref{5DGeodFin} is derived from the $4D$-data $(v,X,Y)$ (according to \eqref{GenVel}). We only have to observe the special form of the null geodesic generators $\gamma^g_{5D}$. In fact, we have 
\begin{equation}\label{eq:gen5d}
    \gamma^g_{5D}(0)=(0,V,Z_p)\quad\text{and consequently}\quad\dot\gamma^g_{5D}(0)=\Big(1,\, -\beta V,\,-\beta Z_i,\,-\beta (Z_4+a)\Big)\,,
\end{equation}
where $ \beta=-(\Lambda/6)\,\Omega V$, and $\Omega=(2a)/(Z_4+a)$.
\medskip

Let us now derive the explicit form of  $T_\eps$. According to \eqref{eq:trsf-overall} we set
\begin{equation}\label{eq:ggwd}
    T_\eps(U,V,Z_p)=\gamma^\eps_{5D}[V,Z_p](U),
\end{equation}
where $\gamma^\eps_{5D}$ solves \eqref{eq:geos} with data
\begin{equation}\label{eq:data-final}
    \gamma^\eps_{5D}(U=-\eps)=\gamma^g_{5D}(-\eps)\quad\text{and}\quad \dot\gamma^\eps_{5D}(U=-\eps)=\dot\gamma^g_{5D}(-\eps)\,,
\end{equation}
i.e., data constructed from the null geodesic generator $\gamma^g_{5D}$ with data \eqref{eq:gen5d} as seed. Since these data essentially reduce to the four parameters $(V,Z_p)$, we will refer to \emph{$\gamma^\eps_{5D}[V,Z_p]$ as the global geodesics with data $(V,Z_p)$.}
Using \eqref{eq:gen5d} we find
\begin{align}
    \gamma^\eps_{5D}(-\eps)&=\Big(-\eps,\, (1+\beta\eps)V,\,(1+\beta\eps)Z_i,\,(1+\beta\eps)Z_4+\beta\eps a\Big),\ \text{and}\\
    \dot\gamma^\eps_{5D}(-\eps)&=\Big(1,\,-\beta V,\,-\beta Z_i,\,-\beta (Z_4+a)\Big)\,.
\end{align}
Now we may write $\gamma^\eps_{5D}[V,Z_p](U)=(\bar U_\eps,\bar V_\eps,\bar Z_{p_\eps})[V,Z_p](U)$ by using \eqref{eq:geos} as 
\begin{equation}
    \gamma^\eps_{5D}[V,Z_p](U)= \gamma^\eps_{5D}[V,Z_p](-\eps)+\dot \gamma^\eps_{5D}[V,Z_p](-\eps)\,(U+\eps)
    +\int\limits_{-\eps}^U\int\limits_{-\eps}^s\, \ddot \gamma^\eps_{5D}[V,Z_p](r) \, dr\,ds\,.
\end{equation}
To do so explicitly we use the following abbreviations for the terms appearing on the r.h.s.\ of \eqref{eq:geos}
\begin{align}\label{eq:Delta}
     \Delta_\eps(r)&:=\frac{1}{2}\, {\dot{\bar{U}}_\eps^2(r)}\, \tilde{G}_\eps(r)-\dot{\bar{U}}_\eps(r)\, 
     \frac{d}{dr}\,\Big(H\big(\bar Z_{p\eps}(r)\big)\,\delta_\eps\big(\bar{U}_\eps(r)\big)\,\bar{U}_\eps(r)\Big)\ \text{with}\\
     \label{eq:Gtilde}
      \tilde{G}_\eps(r)&:=\delta^{pq}\,H_{,p}\big(\bar Z_{w\eps}(r)\big)\,\delta_{\eps}\big(\bar U_\eps(r)\big)\,\bar Z_{q\eps}(r)  + H\big(\bar Z_{w\eps}(r)\big)\,\delta'_\eps\big(\bar U_\eps(r)\big)\,\bar U_\eps(r)\,,\\
     N_\eps(r)&:=\sigma\,a^2-\bar{U}_\eps^2(r)\, H\big(\bar Z_{p\eps}(r)\big)\, \delta_\eps(\bar{U}_\eps(r))\,.\label{eq:N}
\end{align}
With this we find
\begin{align}\label{eq:regtrsf-u}
 \bar U_\eps[V,Z_p](U) &= U - \int\limits_{-\eps}^U\int\limits_{-\eps}^s \frac{\Delta_\eps(r)}{N_\eps(r)}\ \bar U_\eps(r)\ \dd r\dd s\,,\\
 \bar V_\eps[V,Z_p](U)&= (1-\beta U)V 
  +\frac{1}{2}\, \int\limits_{-\eps}^U\int\limits_{-\eps}^s H(\bar Z_{p\eps})(r)\,\delta'_\eps\big(\bar U_\eps(r)\big)\, \dot{\bar{U}}_\eps^2(r)\ \dd r\dd s \nonumber \\
  &\qquad +\int\limits_{-\eps}^U\int\limits_{-\eps}^s\delta^{pq}\, H_{,p}\big(\bar Z_{w\eps}(r)\big)\, \delta_\eps\big(\bar U_\eps(r)\big)\, 
  \dot{\bar Z}_{q\eps}(r)\, \dot{\bar U}_\eps(r)\ \dd r\dd s \\
  &\qquad - \int\limits_{-\eps}^U\int\limits_{-\eps}^s \frac{\Delta_\eps(r)}{N_\eps(r)}\ 
  \Big(\bar V_\eps(r)+H\big(\bar Z_{p\eps}(r)\big)\, \delta_\eps\big(\bar U_\eps(r)\big)\, \bar U_\eps(r)  \Big) 
  \ \dd r\dd s\,,\nonumber\\
 \bar Z_{i\eps}[V,Z_p](U)&= (1-\beta U)Z_i 
   +\frac{1}{2}\ \int\limits_{-\eps}^U\int\limits_{-\eps}^s H_{,i}\big(\bar Z_{p\eps}(r)\big)\, \delta_\eps\big(\bar U_\eps(r)\big)\, \dot{\bar U}_\eps^2(r)\ 
  \dd r\dd s\nonumber\\
  &\qquad - \int\limits_{-\eps}^U\int\limits_{-\eps}^s \frac{\Delta_\eps(r)}{N_\eps(r)}\ \bar Z_{i\eps}(r)\ \dd r\dd s\,,\label{eq:regtrsf-zi}\\
  \bar Z_{4\eps}[V,Z_p](U)&=(1-\beta U)\ \frac{2a}{\alpha} - a
   +\frac{\sigma}{2}\ \int\limits_{-\eps}^U\int\limits_{-\eps}^s H_{,4}\big(\bar Z_{p\eps}(r)\big)\, \delta_\eps\big(\bar U_\eps(r)\big)\, \dot{\bar U}_\eps^2(r)\ 
  \dd r\dd s \nonumber\\
  &\qquad - \int\limits_{-\eps}^U\int\limits_{-\eps}^s \frac{\Delta_\eps(r)}{N_\eps(r)}\ \bar Z_{4\eps}(r)\ \dd r\dd s\,.\label{eq:regtrsf-z4}
  \end{align}
Finally, we observe that the `data parts' of the above equations in fully explicit form read
\begin{align}
(1-\beta U)V&=\Big(1+\frac{\Lambda}{3}\ \frac{a}{Z_4+a}\ V\, U\Big)\ V\,,\\
(1-\beta U)Z_i&=\Big(1+\frac{\Lambda}{3}\ \frac{a}{Z_4+a}\ V\, U\Big)\ Z_i\,,\\
(1-\beta U)\ \frac{2a}{\alpha}-a&=Z_4+a\Big(1+\frac{\Lambda }{3}\ V\, U\Big) - a\,.
\end{align}

\section{Analysis of the regularised transformation}\label{sec:art}
In this section we finally establish that the geometrically regularised transformation $(T_\eps)_\eps$ is a representative of a \emph{generalised diffemorphism} $T=[(T_\eps)_\eps]$ in the sense of nonlinear distributional geometry and thus give a precise mathematical meaning to the physical equivalence of the distributional and the continuous form of the metric. 

The main issue here is of course the subtle interplay between the image of $T_\eps$ and the domain of the inverse. In particular, we have to make sure that the intersection of all images $\cap_{\eps>0}\, \mathrm{im}T_\eps$ contains an open set, which can act as the domain of the inverse. More precisely we use the following definition.

\begin{definition}[generalised diffeomorphism]
Let $\Omega\subseteq \R^n$ be open. We call $T\in\G[\Omega,\R^n]$ a \emph{generalised diffeomorphism} if there exists
$\eta>0$ such that
\begin{enumerate}
 \item[(i)] \label{def-gen-diff-1} There exists a representative $(t_\ep)_\ep$ of $T$ such that $t_\ep\colon\Omega\rightarrow
t_\ep(\Omega)=:\tilde{\Omega}_\ep$ is a diffeomorphism for all $\ep\leq \eta$ and there exists $\tilde{\Omega}\subseteq\R^n$
open with $\tilde{\Omega}\subseteq \bigcap_{\ep\leq\eta}\tilde{\Omega}_\ep$.
 \item[(ii)] The inverses $(t_\ep^{-1})_\ep$ are moderate and c-bounded, i.e., $(t_\ep^{-1})_\ep\in\G[\tilde{\Omega},\R^n]$ and there exists
$\Omega_1\subseteq\R^n$ open, $\Omega_1\subseteq\bigcap_{\ep\leq\eta}t_\ep^{-1}(\tilde{\Omega})$.
 \item[(iii)] Setting $T^{-1}:=[(t_\ep^{-1}|_{\tilde{\Omega}})_\ep]$, the compositions $T\circ T^{-1}$ and $T^{-1}\circ
T|_{\Omega_1}$ are elements of $\G(\tilde{\Omega},\R^n)$ respectively $\G(\Omega_1, \R^n)$. (It is then clear that $T\circ
T^{-1} = id_{\tilde{\Omega}}$ and $T^{-1}\circ T|_{\Omega_1} = id_{\Omega_1})$.
\end{enumerate}
\end{definition}

This definition, of course, extends the smooth theory, cf.\ \cite[Supplement 2.5A]{AMR:88} for a version `quantifying' the neighbourhoods in the classical inverse function theorem. 

We will show that $T$ is a generalised diffeomorphism and we will split up this task in two subsections.

\subsection{$T$ as a generalised function}
To begin with, we have to establish that the regularised transformation $T_\eps$ gives rise to a c-bounded generalised function on $\R^5$, more precisely that $[(T_\eps)_\eps]\in\G[\R^5,\R^5]$. Recall that by \eqref{eq:ggwd} we have $T_\eps(U,V,Z_p)=\gamma^\eps_{5D}[V,Z_p](U)$ and that so far we have only considered $\gamma^\eps_{5D}$ as a function of $U$. Indeed, Theorem \ref{thm:ex+un} guarantees that $[(\gamma^\eps_{5D})_\eps]\in\G[\R,\R^5]$, but now we have to additionally deal with the dependence of $\gamma^\eps_{5D}[V,Z_p](U)$ on $V$ and $Z_p$. 
\medskip

We first establish an appropriate `uniformity of domains' of $\gamma^\eps_{5D}$ in $(V,Z_p)$. To this end we have to delve into the fixed point argument of \cite[Sec.\ III A]{SS:17} that leads to the construction of a local solution candidate for the geodesic equation. Recall from there, or observe from \eqref{eq:geos:G} that the $V$-equation decouples from the system and can simply be integrated after the rest of the system has been solved. So we only have to consider the existence statement \cite[Prop.\ 3.2]{SS:17} which guarantees a local solution for small $\eps$ of a model system which neglects the $V$-equation. There the existence of unique solutions is established on the interval $[\alpha_\eps,\alpha_\eps+\eta]$ when $\eps\leq\eps_0$, where $\eps_0$ and $\eta$ have to satisfy the explicit bounds given in equations \cite[(29), (30)]{SS:17} and the unnumbered equation on top of p.\ 9, respectively. These are explicit bounds in terms of the coefficient functions of the system
(local $L^\infty$-norms of $H$ and $DH$, as well as the $L^1$-norm of $\rho$ and $\rho'$) and the seed data. The latter in our case simplifies to the $Z_p$-components of 
\begin{equation}\label{eq:data-unif-domains}
    \gamma^g_{5D}(0)=(0,V,Z_p),\quad \dot\gamma^g_{5D}(0)=\Big(1,\, -\beta V,\,-\beta Z_i,\,-\beta (Z_4+a)\Big)\,,
\end{equation}
cf.\ \eqref{eq:gen5d}. By inspection it becomes obvious that these estimates can be maintained if the seed data (there $x^0$ and $\dot x^0$) vary in a neighborhood (here of $Z_p$) and that hence $\eta$ and $\eps_0$ can be chosen uniformly on compact neighbourhoods of $Z_p$. Observing that for $\alpha_\eps+\eta\geq\beta_\eps$ the solution again reduces to a background geodesic and using a simple exhaustion argument as in \cite[Prop.\ 4.3]{EG:11}  we obtain the following 
result.
\begin{Lemma}[Uniform domains]
 Given any compact set $K$ in $\R^4$ there is $\eps_0(K)$ such that  $\gamma^\eps_{5D}[V,Z_p]$ is the unique globally defined geodesic of \eqref{eq:ggwd} for all data $(V,Z_p)\in K$ and for all $\eps\leq\eps_0(K)$.
\end{Lemma}

Next we deal with the c-boundedness of $\gamma^\eps_{5D}[V,Z_p]$. Observe that c-boundedness as a function of $U$ is already provided by Theorem \ref{thm:ex+un}, essentially proved in \cite[Prop.\ 4.1 and Appendix A]{SSLP:16}. We now have to see that $\gamma^\eps_{5D}$ is also uniformly bounded if we vary $V$ and $Z_p$ in a compact set. This, however, can also be accomplished by an inspection. For the $Z_p$-components we have to again look into the fixed point argument, more precisely to \cite[Appendix A]{SSLP:16}. The constant $C_2$ of \cite [(A.6)]{SSLP:16} that bounds the solutions again depends on the coefficient functions of the system and the seed data. It clearly can be chosen uniform on compact neighbourhoods of the data, that is of $Z_p$ as the $U$-speed (there $\dot u^0$) in our case is anyways fixed to $1$. So we obtain uniform boundedness of $\bar U_\eps$, and $\bar Z_{p\eps}$ on compact subsets, as well as of the  derivatives $\dot{\bar U}_\eps=\partial_U\bar U_\eps$, and $\dot{\bar Z}_{p\eps}=\partial_U\bar Z_{p\eps}$. Finally, for the $V$-component we have to inspect the boundedness result in \cite[Prop.\ 4.1(iii)]{SSLP:16}. Again it is easily seen that the constants in \cite[(4.1)]{SSLP:16} vary uniformly if $(V,Z_p)$ vary in a compact set. In total we have established the following `uniformity of bounds' result.

\begin{Lemma}[Uniform bounds]\label{lem:ub}
The global geodesics $\gamma^\eps_{5D}[V,Z_p]$  of \eqref{eq:ggwd} are uniformly bounded on compact subsets of $\R^5$ for $\eps$ small enough. In addition such bounds also apply to $\partial_U\bar U_\eps$, and $\partial_U\bar Z_{p\eps}$.
\end{Lemma}

The final task in this subsection is to establish moderateness. We first derive a number of asymptotic estimates which we will also need later on.
\begin{Lemma}\label{lem:NGD} 
Denoting by $\nabla$ any of the derivatives $\partial_V$ and $\partial_{Z_p}$ we have in the regularisation strip $-\eps\leq r\leq\beta_\eps$ and $(V,Z_p)$ varying in a compact set
\begin{align}
    N_\eps(r)&=O(1),\quad \tilde G_\eps(r)=O(\frac{1}{\eps}),\quad 
    \Delta_\eps(r)=O(\frac{1}{\eps})\,, \label{eq:1stest}  \\
    \label{eq:nabla-n}
    \nabla N_\eps &=O(1) \nabla\bar U_\eps + O(\eps)\nabla\bar Z_{p\eps}\,,\\
    \nabla G_\eps(r)&=O(\frac{1}{\eps^2})\nabla\bar U_\eps+O(\frac{1}{\eps})\nabla\bar Z_{p\eps}\,, \\
    \label{eq:nabla-delta}
    \nabla \Delta_\eps &= O(\frac{1}{\eps^2})\nabla \bar U_\eps+
     O(\frac{1}{\eps})(\nabla \bar Z_{p\eps} +\nabla \dot{\bar U}_\eps) + O(1)\nabla\dot{\bar Z}_{p\eps}\,.
\end{align}
\end{Lemma}

\begin{proof}
  The estimates \eqref{eq:1stest} follow direct from the definitions \eqref{eq:Delta}--\eqref{eq:N} observing the boundedness results of Lemma \ref{lem:ub}.
  Similarly we obtain
  \begin{equation}
      \nabla N_\eps(r)=-2\bar U_\eps\nabla \bar U_\eps H\delta_\eps
       -\bar U_\eps^2 DH\nabla Z_{p\eps}\delta_\eps-\bar U_\eps^2 H\delta'_\eps\nabla\bar U_\eps\\
      = O(1) \nabla\bar U_\eps + O(\eps)\nabla\bar Z_{p\eps}\,,
  \end{equation}
  where we have omitted to write out the arguments of $H$, $\delta_\eps$ and the components of $\gamma^\eps_{5D}[V,Z_p]$ explicitly. The result on $\nabla \tilde G_\eps$ simply follows in a similar vein. Finally to derive the estimate on $\Delta_\eps$ first observe that
  \begin{align}
      \frac{d}{dr}(H\delta_\eps\bar U_\eps)&=O(\frac{1}{\eps})\,,\\
      \nabla\frac{d}{dr}(H\delta_\eps\bar U_\eps)&=
      O(\frac{1}{\eps^2})\nabla \bar U_\eps+O(\frac{1}{\eps})(\nabla \bar Z_{p\eps}+\nabla \dot{\bar U}_\eps)+O(1)\nabla\dot{\bar Z}_{p\eps}\,.
  \end{align}
  Now the result again follows along the same lines.
\end{proof}

The next step is to apply Lemma \ref{lem:NGD} to obtain estimates on the first order derivatives of $\bar U_\eps$, $\dot{\bar U}_\eps$, $\bar Z_{p\eps}$, and $\dot{\bar Z}_{p\eps}$. More precisely we have.

\begin{Lemma}[Asymptotic estimates on the first order derivatives]\label{lem:FOD}
We have in the regularisation strip $-\eps\leq r\leq\beta_\eps$ and for $(V,Z_p)$ varying in a compact set 
\begin{eqnarray}
    \partial_V \bar U_\eps=O(\eps^2)\,,\quad
    &\partial_V(\dot{\bar U}_\eps,\bar Z_{p\eps})=O(\eps)\,,\quad 
    &\partial_V\dot{\bar Z}_{p\eps}=O(1)\,, \\
    \partial_{Z_q} \bar U_\eps=O(\eps)\,,\quad
    &\partial_{Z_q}(\dot{\bar U}_\eps,\bar Z_{p\eps})=O(1)\,,\ 
    &\partial_{Z_q}\dot{\bar Z}_{p\eps}=O(1/\eps)\,, 
\end{eqnarray}
as well as
\begin{equation}\label{eq:nabla-n2}
\partial_V N_\eps=O(\eps^2)\,,\quad
\partial_{Z_p} N_\eps=O(\eps)\,,\quad
\partial_V \Delta_\eps=O(1)\,,\quad
\partial_{Z_p} \Delta_\eps=O(1/\eps)\,.
\end{equation}
\end{Lemma}

\begin{proof}
Since here the $\partial_V$- and $\partial_{Z_p}$-derivatives will part ways we introduce the following notation: First we do not distinguish between the individual $\partial_{Z_p}$'s ($p=1,2,3$) and simply write $\partial_Z$.\footnote{In what follows we will estimate the $Z$-derivatives of the data-term in \eqref{eq:regtrsf-zi} simply by $O(1)$, ignoring the fact that $\partial_{Z_i}\bar Z_{j\eps}$ actually gives a term of the form $O(1)\delta_{ij} $.}  Moreover, we will write $\nabla=(\nabla_1,\nabla_2)=(\partial_V,\partial_Z)$ and also use the notation $\nabla_A$ ($A=1,2$).

We aim for a Gronwall estimate using the integral representation of the respective components of the geodesics \eqref{eq:regtrsf-u}, \eqref{eq:regtrsf-zi}, and \eqref{eq:regtrsf-z4}. However, we will do so in a nested way starting with $\nabla \bar U_\eps$, $\nabla \dot{\bar U}_\eps$, and $\nabla \bar Z_{p\eps}$, while leaving $\nabla \dot{\bar Z}_{p\eps}$ for later treatment. 
Setting
\begin{equation}
     \Psi_A:=\max\Big(|\nabla_A \bar U_\eps|, |\nabla_A \dot{\bar U}_\eps|, |\nabla_A \bar Z_{p\eps}|\Big)\,,
\end{equation}
we obtain from \eqref{eq:regtrsf-u} and from \eqref{eq:nabla-n}, \eqref{eq:nabla-delta} 
\begin{align}\label{eq:c1}
   |\nabla_A \bar U_\eps|\,&=\,\int\limits_{-\eps}^{\beta_\eps}\int\limits_{-\eps}^{\beta_\eps}\Big(O\big(\frac{1}{\eps}\big)\,\Psi_A+ O(\eps) |\nabla_A \dot{\bar Z}_{p\eps}|\Big)\,,\\
   |\nabla_A \dot{\bar U}_\eps|\,&=\,\int\limits_{-\eps}^{\beta_\eps}\Big(O\big(\frac{1}{\eps}\big)\,\Psi_A+ O(\eps) |\nabla_A \dot{\bar Z}_{p\eps}|\Big)\,.
\end{align}
Similarly, using \eqref{eq:regtrsf-zi}, \eqref{eq:regtrsf-z4} we obtain by a lengthy calculation
\begin{align}
    |\partial_V \bar Z_{p\eps}|\,&=\, O(\eps)+ \int\limits_{-\eps}^{\beta_\eps}\int\limits_{-\eps}^{\beta_\eps}\Big(O\big(\frac{1}{\eps}\big)\,\Psi_1+ O(\eps) |\partial_V \dot{\bar Z}_{p\eps}|\Big)\,,\\
     |\partial_{Z} \bar Z_{p\eps}|\,&=\, O(1)+ \int\limits_{-\eps}^{\beta_\eps}\int\limits_{-\eps}^{\beta_\eps}\Big(O\big(\frac{1}{\eps}\big)\,\Psi_2+ O(\eps) |\partial_Z \dot{\bar Z}_{p\eps}|\Big)\,.\label{eq:c4}
\end{align}
Summing up we therefore have
\begin{align}
  \Psi_1\,&=\,O(\eps)+ \int\limits_{-\eps}^{\beta_\eps}\Big(O\big(\frac{1}{\eps}\big)\Psi_1+O(\eps)|\partial_V \dot{\bar Z}_{p\eps}|\Big)\,,\\
  \Psi_2\,&=\,O(1)+ \int\limits_{-\eps}^{\beta_\eps}\Big(O\big(\frac{1}{\eps}\big)\Psi_2+O(\eps)|\partial_Z \dot{\bar Z}_{p\eps}|\Big)\,,
\end{align}
and a first appeal to the Gronwall inequality gives
\begin{align}\label{eq:psi1}
    \Psi_1\,=\,O(\eps)\,\Big(1+\int\limits_{-\eps}^{\beta_\eps}|\partial_V \dot{\bar Z}_{p\eps}|\Big)\,, \quad \mbox{and}\quad 
    \Psi_2\,=\,O(1)+O(\eps)\int\limits_{-\eps}^{\beta_\eps}|\partial_Z \dot{\bar Z}_{p\eps}|\,.
\end{align}
Next we turn to $\nabla \dot{\bar Z}_{p\eps}$ for which we find, again from \eqref{eq:regtrsf-zi}, \eqref{eq:regtrsf-z4}, and \eqref{eq:nabla-n}, \eqref{eq:nabla-delta} 
\begin{align}
   |\partial_V \dot{\bar Z}_{p\eps}|\,&=\,
    \int\limits_{-\eps}^{\beta_\eps} O\big(\frac{1}{\eps^2}\big)\Psi_1 +
    \int\limits_{-\eps}^{\beta_\eps} O(1)    |\partial_V \dot{\bar Z}_{p\eps}|
    \,=\,O(1)+ \int\limits_{-\eps}^{\beta_\eps} O(1) |\partial_V \dot{\bar Z}_{p\eps}|\,,\\
    |\partial_Z \dot{\bar Z}_{p\eps}|\,&=\,
    O(1)+\int\limits_{-\eps}^{\beta_\eps} O\big(\frac{1}{\eps^2}\big)\Psi_2 +
    \int\limits_{-\eps}^{\beta_\eps} O(1)    |\partial_Z \dot{\bar Z}_{p\eps}|
    \,=\,O\big(\frac{1}{\eps}\big)+ \int\limits_{-\eps}^{\beta_\eps} O(1) |\partial_Z \dot{\bar Z}_{p\eps}|\,,
\end{align}
where in both lines the second equality follows from \eqref{eq:psi1}. Now a second appeal to the Gronwall inequality hence leaves us with 
\begin{equation}
    |\partial_V \dot{\bar Z}_{p\eps}|=O(1)\,,\quad\text{and}\quad
    |\partial_Z \dot{\bar Z}_{p\eps}|=O\big(\frac{1}{\eps}\big)\,,
\end{equation} 
which already gives the claim on $\nabla \dot{\bar Z}_{p\eps}$.
Also, it allows us to improve the estimates \eqref{eq:psi1} on $\Psi_A$ to
\begin{eqnarray}
    \Psi_1=O(\eps)\,,\quad\mbox{and}\quad \Psi_2=O(1)\,,
\end{eqnarray}
which upon inserting into \eqref{eq:c1}--\eqref{eq:c4} gives the claims on $\partial \bar U_\eps$, $\partial \dot{\bar U}_\eps$, and $\partial \bar Z_\eps$.
Finally, we insert these estimates into \eqref{eq:nabla-n} and \eqref{eq:nabla-delta} to obtain \eqref{eq:nabla-n2}.
\end{proof}


\begin{Lemma}[Asymptotic estimates on the higher order derivatives]\label{lem:HOD}
Denoting by $\nabla$ any of the derivatives $\partial_V$ and $\partial_{Z_p}$ we have in the regularisation strip $-\eps\leq r\leq\beta_\eps$ and $(V,Z_p)$ varying in a compact set that for any $n$ there is $N$ such that
\begin{align}
   \nabla^n(\bar U_\eps, \dot{\bar U}_\eps,\bar Z_{p\eps}, \dot{\bar Z}_{p\eps})\,=\,O(\eps^{-N})\,. 
\end{align}
\end{Lemma}

\begin{proof}
 We proceed by induction. Clearly Lemma \ref{lem:FOD} provides the basis of induction. So assume that we have $\nabla^n(\bar U_\eps, \dot{\bar U}_\eps,\bar Z_{p\eps}, \dot{\bar Z}_{p\eps})=O(\eps^{-M})$ for some $M$.  We again aim for a nested Gronwall argument for the highest order derivatives. Staring with $\bar U_\eps$, and $\dot{\bar U}_\eps$, we find using the integral representation \eqref{eq:regtrsf-u}
 \begin{align}
     \nabla^{n+1}\bar U_\eps\,
     &=\, \int\limits_{-\eps}^{\beta_\eps} O(\eps^2)\nabla^{n+1}\Big(\frac{\Delta_\eps}{N_\eps}\Big)+ 
     \int\limits_{-\eps}^{\beta_\eps} O(1) \nabla^{n+1}\bar U_\eps + O(\eps^{-N})
     \,,\\
     \nabla^{n+1}\dot{\bar U}_\eps\,
     &=\, \int\limits_{-\eps}^{\beta_\eps} O(\eps)\nabla^{n+1}\Big(\frac{\Delta_\eps}{N_\eps}\Big)+ 
     \int\limits_{-\eps}^{\beta_\eps} O\big(\frac{1}{\eps}\big) \nabla^{n+1}\bar U_\eps + O(\eps^{-N})
     \,,
 \end{align}
 where we have only retained the highest order terms explicitly and estimated all lower order terms by some large inverse power of $\eps$. Next we deal with the term $\nabla^{n+1}\bar Z_{p\eps}$ for which we find from \eqref{eq:regtrsf-zi}, and \eqref{eq:regtrsf-z4}
 \begin{equation}\label{eq:est-Zp}
    \nabla^{n+1}\bar Z_{p\eps}\,=\,
    \int\limits_{-\eps}^{\beta_\eps} \left(O\big(\frac{1}{\eps}\big)\,\Big(\nabla^{n+1}\bar U_{\eps}+\nabla^{n+1}\dot{\bar U}_{\eps}+\nabla^{n+1}\bar Z_{p\eps}\Big) + O(\eps)\nabla^{n+1}\Big(\frac{\Delta_\eps}{N_\eps}\Big)\right) + O(\eps^{-N})\,,
 \end{equation}
 simply collecting all lower order terms in the final $O(\eps^{-N})$-estimate.
 Observe that the (critical) term involving $\delta_\eps(\bar U_\eps)$ does not produce any (high) inverse powers of $\eps$ in the highest order terms $\nabla^{n+1}(\bar U_{\eps},\dot{\bar U}_{\eps},\bar Z_{p\eps})$.

 Now we set $\Psi=\max(|\nabla^{n+1}\bar U_{\eps}|,|\nabla^{n+1}\dot{\bar U}_{\eps}|,|\nabla^{n+1}\bar Z_{p\eps}|)$ and obtain by collecting the above estimates
 \begin{equation} \label{eq:psi-int}
      \Psi\,=\, \int\limits_{-\eps}^{\beta_\eps} \left( O(\eps)\nabla^{n+1}\Big(\frac{\Delta_\eps}{N_\eps}\Big)+O\big(\frac{1}{\eps}\big)\Psi\right)+O(\eps^{-N})\,,
\end{equation}
 and so by a first application of Gronwall's inequality
 \begin{equation}\label{eq:psi-noint}
     \Psi\,=\, O(\eps^2)\nabla^{n+1}\Big(\frac{\Delta_\eps}{N_\eps}\Big) + O(\eps^{-N})\,.
 \end{equation}
 Next we use the integral representations to obtain the following estimate on the $(n+1)$-st derivative of the fraction $\Delta_\eps/N_\eps$
 \begin{align}\nonumber
     \nabla^{n+1}\Big(\frac{\Delta_\eps}{N_\eps}\Big)\,
     &=\, O\big(\frac{1}{\eps^2}\big)\nabla^{n+1}\bar U_\eps+
      O\big(\frac{1}{\eps}\big)\Big(\nabla^{n+1}\dot{\bar U}_\eps+\nabla^{n+1} \bar Z_{p\eps}\Big)+O(1)\nabla^{n+1} \dot{\bar Z}_{p\eps}+O(\eps^{-N})\\
      &=\, \int\limits_{-\eps}^{\beta_\eps} O\big(\frac{1}{\eps}\big) \nabla^{n+1}\Big(\frac{\Delta_\eps}{N_\eps}\Big)\ +O(1)\nabla^{n+1}\dot{\bar Z}_{p\eps}+O(\eps^{-N})\,,
 \end{align}
 where we have used the estimates \eqref{eq:psi-int}, \eqref{eq:psi-noint}. So another appeal to Gronwall's inequality yields
 \begin{equation}\label{eq:est-quot}
     \nabla^{n+1}\Big(\frac{\Delta_\eps}{N_\eps}\Big)\,=\,
     O(1)\nabla^{n+1}\dot{\bar Z}_{p\eps}+O(\eps^{-N})\,.
 \end{equation}
 Inserting this back into the $\Psi$-estimate \eqref{eq:psi-noint} we find
 \begin{equation}\label{eq:psi-final}
  \Psi\,=\, O(\eps^2)\nabla^{n+1}\dot{\bar Z}_{p\eps} + O(\eps^{-N})\,.
 \end{equation}
 Finally, we turn to the term $\nabla^{n+1}\dot{\bar Z}_{p\eps}$, for which we find again from the integral representation \eqref{eq:regtrsf-zi}, and \eqref{eq:regtrsf-z4} (cf.\ \eqref{eq:est-Zp})
 \begin{align}\nonumber
     \nabla^{n+1}\dot{\bar Z}_{p\eps}\,&=\,
      \int\limits_{-\eps}^{\beta_\eps} \left(
       O\big(\frac{1}{\eps^2}\big)\Big(\nabla^{n+1}\bar U_{\eps}+\nabla^{n+1}\dot{\bar U}_{\eps}+\nabla^{n+1}\bar Z_{p\eps}\Big) 
       +O(1)\nabla^{n+1}\Big(\frac{\Delta_\eps}{N_\eps}\Big)\right)
       +O(\eps^{-N})
     \\
     &=\int\limits_{-\eps}^{\beta_\eps}  O(1)\nabla^{n+1}\dot{\bar Z}_{p\eps}+O(\eps^{-N})\,,
 \end{align}
 where we have used \eqref{eq:psi-final} as well as \eqref{eq:est-quot}. So a final appeal to Gronwall's estimate gives $ \nabla^{n+1}\dot{\bar Z}_{p\eps}=O(\eps^{-N})$, and, upon inserting into \eqref{eq:psi-final}, $\Psi=O(\eps^{-N})$, which is the claim.
\end{proof}

Now we finally obtain moderateness of the transformation. Indeed, we have the following more specific result. 

\begin{Proposition}[Moderateness of the transformation]\label{prop:mod}
The net of transformations $(T_\eps)_\eps$ is moderate and hence $[(T_\eps)_\eps]$ is an element of $\G[\R^5,\R^5]$. 
\end{Proposition}

\begin{proof}
Recall that we only have to argue inside the regularised wave zone $-\eps\leq\bar U_\eps\leq\beta_\eps$ since outside of it $\gamma^\eps_{5D}$ coincides with classical smooth solutions (depending smoothly on $\eps$). The c-boundedness (in this strip) was established in Lemma \ref{lem:ub} and the moderateness estimates for $\partial_U\gamma^\eps_{5D}[V,Z_p](.)$ are due to Theorem \ref{thm:ex+un}. The moderateness estimates for the $V$- and $Z_q$-derivatives of $\bar U_\eps$, and $\bar Z_{p\eps}$, as well as their mixed $V$-$Z_q$-derivatives have been established in Lemma \ref{lem:HOD}. The mixed $U$-$V$-$Z_q$-derivatives follow suit by iteratively using the differential equation\footnote{Observe that it was precisely the tricky point in the proofs above that one cannot use the differential equation beyond the integral representation when estimating the $V$-$Z_q$-derivatives.}.

Finally, since the $V$-equation is decoupled from the system and $\bar V[V,Z_p](.)$ is obtained simply by integration of the other components its moderateness is a consequence of moderateness of $(\bar U,\bar Z_p)$ (and the well-definedness of the respective operations in $\G[\R^n,\R^m]$). 
\end{proof}

\subsection{$T$ as generalised diffeomorphism}

We will now set out to show that the transformation \eqref{eq:ggwd}, i.e.,   
\begin{equation}
    (U,V,Z_p)\mapsto[T_\eps](U,V,Z_p)=[\gamma^\eps_{5D}][V,Z_p](U)\, \in \G[\R^5,\R^5]\,,
\end{equation}
with its components given explicitly by \eqref{eq:regtrsf-u} - \eqref{eq:regtrsf-z4} gives rise to a locally invertible generalised function $T=[(T_\eps)_\eps]$ on some open set containing the impulsive surface. Hence we will  call it a \emph{generalised diffeomorphism} or \emph{generalised coordinate transformation}. To do so we will extend the results of the $\Lambda=0$-case of \cite{KS:99} and, in particular, its more mathematically structured presentation in \cite{EG:11}. In fact, inspired by \cite{EG:11} we will decompose the transformation in a convenient way by splitting $T_\eps$ into a `singular' and a `convergent' part. 


To begin with, fix an open, relatively compact set $W\subseteq\R^5$, which will be specified further later, and observe that by Proposition \ref{prop:mod} $(T_\eps)$ is moderate and c-bounded and therefore indeed $T:=[(T_\eps)_\eps]\in\G[W,\R^n]$.
%
We decompose the $\bar V_\eps$-component into the initial data term $\tilde V_\eps$ and the integral term, which we label as $h_\eps$, so that we may write 
\begin{equation}
    \bar V_\eps = \tilde V_\eps + h_\eps\,.
\end{equation}
To be precise, we have 
\begin{align*}
    \tilde V_\eps &=  (1-\beta U)V \,,\\
    h_\eps &=  \frac{1}{2}\, \int\limits_{-\eps}^U\int\limits_{-\eps}^s H(\bar Z_{p\eps})(r)\,\delta'_\eps\big(\bar U_\eps(r)\big)\, \dot{\bar{U}}_\eps^2(r)\ \dd r\dd s \nonumber \\
  &\qquad +\int\limits_{-\eps}^U\int\limits_{-\eps}^s\delta^{pq}\, H_{,p}\big(\bar Z_{w\eps}(r)\big)\, \delta_\eps\big(\bar U_\eps(r)\big)\, 
  \dot{\bar Z}_{q\eps}(r)\, \dot{\bar U}_\eps(r)\ \dd r\dd s \\
  &\qquad - \int\limits_{-\eps}^U\int\limits_{-\eps}^s \frac{\Delta_\eps(r)}{N_\eps(r)}\ 
  \Big(\bar V_\eps(r)+H\big(\bar Z_{p\eps}(r)\big)\, \delta_\eps\big(\bar U_\eps(r)\big)\, \bar U_\eps(r)  \Big) 
  \ \dd r\dd s\,.
\end{align*}
Note that while $h_\eps$ does not converge, we have that $h_\eps =O(1)$, cf.\ the proof of Proposition 4.1 in \cite{SSLP:16} and in particular equations (4.3) and (4.4). At this point we define the converging sequence $s_\eps(U,V,Z):=(\bar U_\eps, \tilde V_\eps, \bar Z_{i\eps},\bar Z_{4\eps})$. 


We will use \cite[Prop.\ 3.16 and Thm.\ 3.59]{Erl:07} in order to establish injectivity of $T_\eps$, and therefore we need to find the asymptotic behaviour of $DT_\eps$ and $Ds_\eps$. More specifically, we need to estimate the behaviour of their determinant and all principal minors. The Jacobian of $s_\eps$ is 

\begin{equation}\label{eq-Dseps}
  Ds_\eps=  \left( \begin{array}{ccccc}
       1&0&0&0  \\
       \frac{\Lambda}{3} \frac{a}{Z_4+a} V^2 &1+\frac{2\Lambda}{3} \frac{a}{Z_4+a} U V&0&-\frac{\Lambda}{3} \frac{a}{(Z_4+a)^2} U V^2 \\
       \frac{\Lambda}{3} \frac{a}{Z_4+a} V Z_i & \frac{\Lambda}{3} \frac{a}{Z_4+a} U Z_i &1+ \frac{\Lambda}{3}\frac{a}{Z_4+a} U V 
       & -\frac{\Lambda}{3}  \frac{a}{(Z_4+a)^2} U V Z_i\\
       \frac{\Lambda}{3} a V & \frac{\Lambda}{3} a U&0&1\\
    \end{array}\right) + B\,,
\end{equation}
where $B$ is a matrix with all entries $O(\eps)$, as we will see in Lemma \ref{lem-asy} below.
For the following we denote 
\begin{equation}
\Omega:= \{ (U,V,Z)\in \R^5:\, \left|\frac{\Lambda}{3} \frac{a}{a+Z_4} V U\right|\leq\frac{1}{8}\}\,.
\end{equation}

\begin{Lemma}
    There is $\eps_0>0$ such that on any closed rectangular subset of $\Omega$ and for all $\eps< \eps_0$, $s_\eps$ is injective.
\end{Lemma}

\begin{proof}
For $\eps$ small enough we only have to consider the matrix $Ds_\eps-B$. Condition \eqref{eq-Dseps} guarantees that all principal minors are bounded below by a fixed positive constant on all of $\Omega$. Thus \cite[Thm.\ 4]{GN:65} gives the claimed injectivity.
\end{proof}

Outside the regularisation strip $\{\alpha_\eps \leq U \leq \beta_\eps\}$ the transformation defaults to a smooth coordinate transform independent of $\eps$ and hence possesses all the properties needed in the following arguments. Thus we may restrict ourselves to the regularisation zone and there $\bar U_\eps = O(\eps)$ holds. As for $T_\eps$, being able to decompose its $V$ component into $\tilde V_\eps + h_\eps$ lets us more easily compute its Jacobian.
First, we focus on the regularisation strip $\{\alpha_\eps \leq U \leq \beta_\eps\}$:

\begin{Lemma}[Asymptotics estimates for $\bar V_\eps$]\label{lem-asy}
We have in the regularisation strip $-\eps\leq r\leq\beta_\eps$ and $(V,Z_p)$ varying in a compact set 
\begin{align} 
\bar V_\eps = O(1),\  
\pt_V \bar V_\eps =O(1), \  \pt_V h_\eps = O(\eps),\  
\pt_{Z_{p}}h_\eps=O(1)\,.
\end{align}
\end{Lemma}
\begin{proof}
    The asymptotics of $\bar V_\eps$ have already been noted in Proposition \ref{prop:mod}.
We estimate the $V$-component: As $h_\eps$ consists of three integrals, we split up the calculation, writing $h_\eps = \frac{1}{2}I_1+I_2-I_3$. Then using the above we obtain that
\begin{align*}
        \pt_V I_1 =& \int_{-\eps}^U \int_{-\eps}^s DH(\bar Z_{p\eps}) \pt_V \bar Z_{p\eps} \delta'(\bar U_\eps) \dot{\bar U}_\eps^2 + H(\bar Z_{p\eps}) \delta''(\bar U_\eps) \pt_V \bar U_\eps \dot{\bar U}_\eps^2 
        + H(\bar Z_{p\eps}) \delta'(\bar U_\eps) 2 \dot{\bar U}_\eps \pt_V \dot{\bar  U}_\eps \\
        =&\int_{-\eps}^U \int_{-\eps}^s O(1) O(\eps) O(\frac{1}{\eps^2}) O(1) + O(1) O(\frac{1}{\eps^3}) O(\eps^2) O(1)+ O(1) O(\frac{1}{\eps^2}) O(1) O(\eps)\\
         =& \int_{-\eps}^U \int_{-\eps}^s O(\frac{1}{\eps}) = O(\eps)\,.
\end{align*}
Similarly, we obtain that
\begin{align}
    \pt_V I_2 = 
    \int_{-\eps}^U \int_{-\eps}^s O(\frac{1}{\eps}) =O(\eps)\,.
\end{align}
For the last term we need to be more careful. First we calculate
\begin{align*}
        \pt_V I_3 &= \int_{-\eps}^U \int_{-\eps}^s \frac{\pt_V \Delta_\eps \bigl(\bar V_\eps + H(\bar Z_{p\eps}) \delta(\bar U_\eps) \bar U_\eps\bigr)}{N_\eps}\\
        &+ \frac{\Delta_\eps\bigl(\pt_V \bar V_\eps +  DH(\bar Z_{p\eps}) \pt_V \bar Z_{p\eps} \delta_\eps(\bar U_\eps)\bar U_\eps + H(\bar Z_{p\eps}) \delta_\eps'(\bar U_\eps) \pt_V \bar U_\eps \bar U_\eps + H(\bar Z_{p\eps}) \delta_\eps(\bar U_\eps) \pt_V \bar  U_\eps\bigr)}{N_\eps}\\
        &- \frac{\Delta_\eps \bigl(\bar V_\eps + H(\bar Z_{p\eps})\delta_\eps(\bar U_\eps)\bar U_\eps\bigr)\pt_V N_\eps}{N_\eps^2} \\
        &=\int_{-\eps}^U \int_{-\eps}^s O(\frac{1}{\eps}) + O(\frac{1}{\eps}) \pt_V h_\eps = O(\eps) + \int_{-\eps}^U \int_{-\eps}^s O(\frac{1}{\eps}) \pt_V h_\eps\,.
\end{align*}
Consequently, we get that $\pt_V h_\eps = O(\eps)$ and so $\pt_V \bar V_\eps = O(1)$ as the $V$-derivative of the initial conditions for the $V$-component is $O(1)$.
\end{proof}

At this point we observe that 
\begin{align*}
  DT_\eps&=\\
  &\left( \begin{array}{ccccc}
       1&0&0&0  \\
       \frac{\Lambda}{3} \frac{a}{Z_4+a} V^2 &1+\frac{2\Lambda}{3} \frac{a}{Z_4+a} U V&0&-\frac{\Lambda}{3} \frac{a}{(Z_4+a)^2} U V^2 \\
       \frac{\Lambda}{3} \frac{a}{Z_4+a} V Z_i & \frac{\Lambda}{3} \frac{a}{Z_4+a} U Z_i &1+ \frac{\Lambda}{3}\frac{a}{Z_4+a} U V 
       & -\frac{\Lambda}{3}  \frac{a}{(Z_4+a)^2} U V Z_i\\
       \frac{\Lambda}{3} a V & \frac{\Lambda}{3} a U&0&1\\
    \end{array}\right)\\
    &+
    \left( \begin{array}{ccccc}
       0&0&0&0  \\
      \pt_U h_\eps & 0&\pt_{Z_i} h_\eps & \pt_{Z_4} h_\eps \\
       \pt_U \tilde{Z}_{i\eps} & 0 & 0 & 0\\
       \pt_U \tilde{Z}_{4\eps} & 0&0&1\\
    \end{array}\right) + \tilde B\\
    &= Ds_\eps + \left( \begin{array}{ccccc}
       0&0&0&0  \\
      \pt_U h_\eps & 0&\pt_{Z_i} h_\eps & \pt_{Z_4} h_\eps \\
       \pt_U \tilde{Z}_{i\eps} & 0 & 0 & 0\\
       \pt_U \tilde{Z}_{4\eps} & 0&0&1\\
    \end{array}\right)+ \tilde B =: Ds_\eps + I_\eps + \tilde B\,,
\end{align*}
where the $\tilde{Z}_{q\eps}$ are the $\bar{Z}_{q\eps}$ without the initial conditions, $\tilde B$ comes from $B$ in \eqref{eq-Dseps} and the additional $\pt_V h_\eps$-, $\pt_V \bar{Z}_{q\eps}$-, $\pt_{Z_p} \bar{Z}_{q\eps}$-terms are all $O(\eps)$ by Lemma \ref{lem-asy}. Furthermore, the $(1,2)$-entry of $\tilde B$, which is $\pt_V \bar{U}_\eps$ is even $O(\eps^2)$ (by Lemma \ref{lem-asy}), which is essential in what follows. Thus when calculating all the principal minors of $DT_\eps$ 
we need to observe that
\begin{enumerate}
    \item the factor $\pt_U h_\eps$, which is $O(\frac{1}{\eps})$, is always multiplied by an $O(\eps^2)$-term,
    \item the factors $\pt_U \tilde{Z}_{q\eps}$, which are $O(1)$, are always multiplied by an $O(\eps)$-term, and
    \item the factors $\pt_{Z_p} h_\eps$, which are $O(1)$ by Lemma \ref{lem-asy}, are always multiplied by an $O(\eps)$-term.
\end{enumerate}
Thus, all the principal minors are of the form $1 + O(\eps)$, and hence, in particular, $|\det(DT_\eps)|\geq \eps^N$ for some $N\in\N$. Consequently, $(T_\eps^{-1})_\eps$ is moderate, and from $T_\eps^{-1} \circ T_\eps = \mathrm{id}$ we conclude that $T_\eps^{-1}$ is c-bounded (on the image of $T_\eps$). 

In conclusion, this gives that $T=[(T_\eps)_\eps)]$ is a generalised diffeomorphism. We are, however, interested in the overall transformation \eqref{eq:trsf-overall}, i.e., the precomposition of $T$ with \eqref{CoordTrans_4D_to_5D} and the postcomposition with \eqref{CoordTrans_5D_to_4D}. Since they both are (classical) smooth diffeomorphsims on their respective domains we only need to observe that such a composition clearly is a generalised diffeomorphism. So in total we have:

\begin{Theorem}
    The discontinuous coordinate transform \eqref{eq:trsf-overall} is a generalised diffeomorphism.
\end{Theorem}

\section{Discussion}\label{sec:6}
In this work we have studied the notorious discontinuous coordinate transformation  \eqref{trans} relating the distributional and the continuous metric commonly used to describe non-expanding impulsive gravitational waves propagating in (anti-)de Sitter space. Already in \cite{PSSS:19} it was shown that this transformation is geometrically given by the null generators of the (A)dS hyperboloid in a $5D$-description, which jump and are refracted due to the wave impulse. Here we have put this formal analysis on firm mathematical grounds using the nonlinear distributional analysis of the geodesics in these geometries provided in \cite{SSLP:16,SS:18}. More precisely, we have established that a careful geometric regularisation of the transformation leads to a generalised diffeomorphism in the sense of nonlinear distributional geometry. In this way we have also generalised the 
analysis of the far simpler $\Lambda=0$-case of \cite{KS:99a}. We have schematically displayed our procedure in Figure \ref{fig:big-diag}.
\medskip

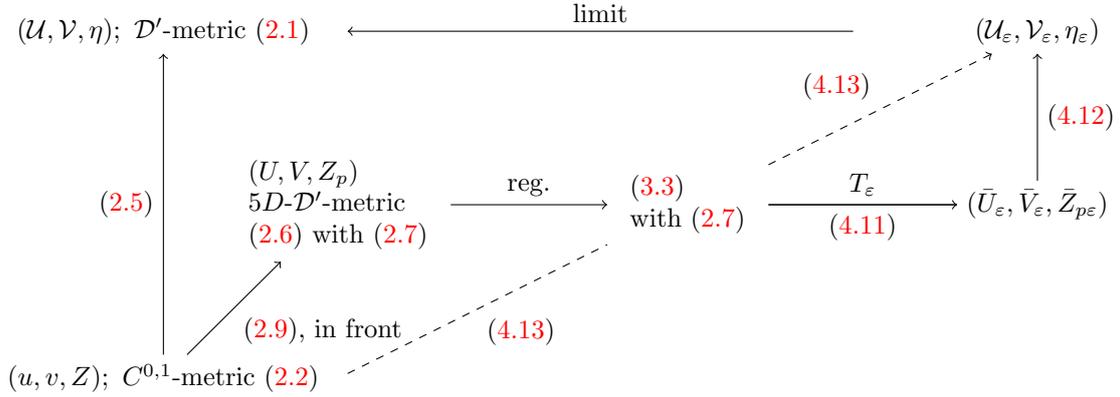
\begin{figure}
 \centering
\begin{tikzpicture}[node distance=2.3cm, auto]
  \node (11) {$({\mathcal U},{\mathcal V},\eta);\ {\mathcal D}'\text{-metric}\  \eqref{4D_imp}$};
  \node (21)  [below of=11] {};
  \node (31) [below of=21] {$(u,v,Z);\ C^{0,1}\text{-metric}\  \eqref{conti}$};
  \node (32) [right of=31] {};
  \node (12) [right of=11] {};
  \node (22) [below of=12] {$\begin{array}{l}(U,V,Z_p)\\ 5D\text{-}{\mathcal D}'\text{-metric}\\ \eqref{5D_imp}\ \text{with}\ \eqref{Constraint_Hyp}\end{array}$};
  \node (23) [right of=22] {};
  \node (24) [right of=23] {$\begin{array}{l}\eqref{5ipp}\\ \text{with}\ \eqref{Constraint_Hyp}\end{array}$};
  \node (25) [right of=24] {};
  \node (26) [right of=25] {$(\bar U_\eps,\bar V_\eps,\bar{Z}_{p\eps}) $};
  \node (16) [above of=26] {$({\mathcal U}_\eps,{\mathcal V}_\eps,\eta_\eps)$};
  \node (15) [left of=16] {};
  %
  \draw[->] (31) to node {\eqref{trans}} (11);
  \draw[->] (31) to node [swap] {\eqref{CoordTrans_4D_to_5D}, in front} (22);
  \draw[->] (22) to node {\text{reg.}} (24);
  \draw[->] (24) to node {$T_\eps$}  (26);
  \draw[ ]  (24) to node [swap] {\eqref{eq:regtrsf}} (26);
  \draw[->] (26) to node [swap]{\eqref{CoordTrans_5D_to_4D}} (16);
  \draw[->] (15) to node [swap] {\text{limit}} (12);
  \draw[-,dashed] (32) to node [swap] {\eqref{eq:trsf-overall}} (24);
  \draw[->,dashed] (24) to node {\eqref{eq:trsf-overall}} (16);
\end{tikzpicture}
\caption{Schematics of the transformations and regularisations employed in this work.}
\label{fig:big-diag}
\end{figure}

Physically speaking our approach consists in viewing the impulsive wave as a limiting case of a sandwich wave, where we have used the $5D$-formalism to define at a sensible regularisation of the spacetime, i.e., as the (A)dS hyperboloid in a $5D$ flat sandwich wave. From this point of view the two forms of the impulsive metric arise as the (distributional) limits of this sandwich
wave in different coordinate systems, once in the $4D$-`continuous system' $(u,v,Z)$, where the metric is \eqref{conti} and the $4D$-`distributional system' $({\mathcal U},{\mathcal V},\eta)$, where the metric is \eqref{4D_imp}.

\subsection*{Acknowledgment}
The authors want to thank Ji\v{r}\'{i} Podolsk\'y for constantly sharing his expertise. C.S., B.S., and R.S.\ acknowledge the kind hospitality of the Erwin Schrödinger International Institute for Mathematics and Physics (ESI) during the workshop Nonregular Spacetime Geometry, where parts of this research were carried out. 

This research was funded in part by the Austrian Science Fund (FWF) [Grant DOI 10.55776/P\-33594]. For open access purposes, the authors have applied a CC BY public copyright license to any author accepted manuscript version arising from this submission. C.S.\ is supported by the European Research Council (ERC), under the European’s Union Horizon 2020 research and innovation programme, via the ERC Starting Grant “CURVATURE”, grant agreement No.\ 802689. R.\v{S}.\ acknowleges support from the Czech Science Foundation Grant No.\ GA{\v C}R 22-14791S.


\subsection*{Data availability satement} Not applicable to this article because it is based on purely theoretical considerations,
without using any datasets or other materials.

\subsection*{Conflict of interest statement}
On behalf of all authors, the corresponding author (R.S.) states that there is no conflict of interest.

\bibliographystyle{abbrv}
\bibliography{ro}

\end{document}